\pdfoutput=1

\documentclass[11pt,twoside,a4paper,notitlepage]{article}
\usepackage[a4paper,margin=2.5cm]{geometry}

\makeatletter{}%

\usepackage[utf8]{inputenc}
\usepackage{color}
\usepackage{tikz}

\usepackage[noend]{algpseudocode}

\usepackage{xspace}
\usepackage{bm}

\usepackage{algorithm}
\usepackage{hyperref}

\usepackage{amsthm}
\usepackage{amsmath} 
\usepackage{amsfonts} 
\usepackage{amssymb}   
\usepackage{stmaryrd} 
\usepackage{enumerate}

\newtheorem{theorem}{Theorem}[section] 
\newtheorem{lemma}[theorem]{Lemma}

\theoremstyle{definition}
\newtheorem{definition}[theorem]{Definition}
\newtheorem{claim}[theorem]{Claim}

\newtheorem{example}[theorem]{Example}

\newcommand{\nc}[1]{\newcommand{#1}}
\newcommand{\rnc}[1]{\renewcommand{#1}}

\rnc{\leq}{\ensuremath{\leqslant}}
\rnc{\geq}{\ensuremath{\geqslant}}

\rnc{\le}{\leq}
\rnc{\ge}{\geq}

\nc{\isdef}{\ensuremath{:=}}
\nc{\deff}{\isdef}
\nc{\defi}{\isdef}

\nc{\set}[1]{\ensuremath{\{#1\}}}
\nc{\setsize}[1]{\ensuremath{|#1|}}
\nc{\Setsize}[1]{\ensuremath{\big|#1\big|}}
\nc{\Set}[1]{\ensuremath{\big\{#1\big\}}}
\nc{\setc}[2]{\set{#1 \ : \ #2}}
\nc{\Setc}[2]{\Set{#1 \ : \ #2}}

\nc{\aufgerundet}[1]{\ensuremath{\lceil #1 \rceil}}
\nc{\abgerundet}[1]{\ensuremath{\lfloor #1 \rfloor}}

\nc{\dcup}{\ensuremath{\dot\cup}}

\nc{\ov}[1]{\ensuremath{\overline{#1}}}

\nc{\NN}{\ensuremath{\mathbb{N}}}
\nc{\NNpos}{\ensuremath{\NN_{\scriptscriptstyle\geq 1}}}
\nc{\RR}{\ensuremath{\mathbb{R}}}
\nc{\RRpos}{\ensuremath{\RR_{\scriptscriptstyle\geq 0}}}

\nc{\und}{\ensuremath{\wedge}}
\nc{\Und}{\ensuremath{\bigwedge}}
\nc{\oder}{\ensuremath{\vee}}
\nc{\Oder}{\ensuremath{\bigvee}}
\nc{\nicht}{\ensuremath{\neg}}
\nc{\impl}{\ensuremath{\to}}
\nc{\gdw}{\ensuremath{\leftrightarrow}}

\newcommand{\uund}{\,\und\,}

\newcommand{\bbody}[1]{\;\big(\,#1\,\big)}

\nc{\free}{\ensuremath{\textrm{\upshape free}}}
\nc{\ar}{\ensuremath{\operatorname{ar}}}

\nc{\Structure}[1]{\ensuremath{\mathcal{#1}}}
\nc{\A}{\Structure{A}}
\nc{\B}{\Structure{B}}
\nc{\C}{\Structure{C}}

\nc{\isom}{\ensuremath{\cong}}

\nc{\querycont}{\ensuremath{\sqsubseteq}}

\nc{\eval}[2]{\ensuremath{#1(#2)}}
\nc{\semantik}[1]{\ensuremath{\left\llbracket#1\right\rrbracket}}

\nc{\CanDB}[1]{\ensuremath{\A_{#1}}} 
\nc{\CanTup}[1]{\ensuremath{t_{#1}}} 

\newcommand{\tupleVariables}{\ensuremath{\ov{v}}}

\newcommand{\inds}{s}
\newcommand{\queryphi}{\varphi}
\newcommand{\varv}{v}

\newcommand{\varx}{x}
\newcommand{\vary}{y}
\newcommand{\varz}{z}
\newcommand{\varw}{w}

\newcommand{\sgpsi}{\psi} 
 
\newcommand{\relS}{S} 
 
\newcommand{\relT}{T} 

\newcommand{\relE}{E} 

\newcommand{\relR}{R} 
\newcommand{\smalleps}{\varepsilon}
\newcommand{\arityr}{r}

\newcommand{\actdomsize}{n}

\newcommand{\verta}{a} 
\newcommand{\vertb}{b} 
\newcommand{\vertc}{c}

\nc{\Vars}{\ensuremath{\textrm{\upshape vars}}}
\nc{\vars}{\Vars}
\nc{\Cons}{\ensuremath{\textrm{\upshape cons}}}
\nc{\cons}{\Cons}
\nc{\atoms}{\ensuremath{\textrm{\upshape atoms}}}

\nc{\Adom}{\ensuremath{\textrm{\upshape adom}}}
\nc{\adom}[1]{\ensuremath{\Adom(#1)}} 
\nc{\dom}[1]{\ensuremath{\textrm{\upshape dom}(#1)}} 

\newcommand{\poly}{\operatorname{\textit{poly}}}

\newcommand{\qhier}{q-hie\-rar\-chi\-cal\xspace}

\newcommand{\OMv}{OMv\xspace}
\newcommand{\OMvcon}{\OMv{}-conjecture\xspace}

\newcommand{\OV}{OV\xspace}
\newcommand{\OVcon}{\OV-conjecture\xspace}

\newcommand{\indi}{i}
\newcommand{\indj}{j}
\newcommand{\indt}{t} 

\newcommand{\vecu}{\vec{u}}

\newcommand{\vecv}{\vec{v}}

\newcommand{\matM}{M}
\newcommand{\dimn}{n}

\newcommand{\DBone}[1]{}

\newcommand{\bigoh}{O}
\newcommand{\bigOh}{\bigoh}
\newcommand{\trans}{^{\,\mkern-1.5mu\mathsf{T}}}

\newcommand{\parent}{\pointerfont{parent}}

\nc{\arrayfont}[1]{\ensuremath{\texttt{#1}}}

\newcommand{\query}{\ensuremath{\varphi}}
\newcommand{\qatom}{\ensuremath{\psi}}

\newcommand{\size}[1]{\ensuremath{|\!|#1|\!|}}
\nc{\card}[1]{\ensuremath{|#1|}}

\newcommand{\assign}{\ensuremath{\alpha}}

\newcommand{\potenzmengeof}[1]{2^{#1}}

\newcommand{\sgpsix}{\sgpsi^\varx}
\newcommand{\sgpsiy}{\sgpsi^\vary}
\newcommand{\sgpsixy}{\sgpsi^{\varx,\vary}}

\newcommand{\partPfull}{\mathcal P}

\newcommand{\iotasubij}{\iota_{\indi,\indj}} 

\newcommand{\homh}{h}

\newcommand{\homDBtoquery}{g}

\nc{\insertp}{\textsc{Insert}}
\nc{\cleanup}{\textsc{cleanUp}}
\nc{\cleanups}{\textsc{cleanUp'}}

\nc{\myparagraph}[1]{\medskip\noindent\textbf{#1}. }

\nc{\Yes}{\texttt{yes}}
\nc{\No}{\texttt{no}}

\nc{\Dom}{\ensuremath{\textbf{dom}}}
\nc{\Var}{\ensuremath{\textbf{var}}}
\nc{\schema}{\ensuremath{\sigma}}
\nc{\DB}{\ensuremath{D}} 
\nc{\DBstrich}{\ensuremath{D'}} 
\nc{\DBstart}{\ensuremath{{\DB_0}}} 
                                
\nc{\DBempty}{\ensuremath{{\DB_{\emptyset}}}} 

\nc{\DS}{\ensuremath{\mathtt{D}}} 

\rnc{\phi}{\queryphi}

\nc{\UpdateFont}[1]{\ensuremath{\textsf{#1}}}
\nc{\Delete}{\UpdateFont{delete}}
\nc{\Insert}{\UpdateFont{insert}}
\nc{\Update}{\UpdateFont{update}}

\nc{\AlgoFont}[1]{\ensuremath{\textbf{#1}}}
\nc{\PREPROCESS}{\AlgoFont{preprocess}}
\nc{\INIT}{\AlgoFont{init}}
\nc{\UPDATE}{\AlgoFont{update}}
\nc{\ENUMERATE}{\AlgoFont{enumerate}}
\nc{\COUNT}{\AlgoFont{count}}
\nc{\ANSWER}{\AlgoFont{answer}}
\nc{\TEST}{\AlgoFont{test}}

\nc{\EOE}{\texttt{\upshape EOE}\xspace} 

\nc{\preprocessingtime}{\ensuremath{t_p}}
\nc{\inittime}{\ensuremath{t_i}}
\nc{\delaytime}{\ensuremath{t_d}}
\nc{\updatetime}{\ensuremath{t_u}}
\nc{\answertime}{\ensuremath{t_a}}
\nc{\countingtime}{\ensuremath{t_c}}
\nc{\testingtime}{\ensuremath{t_t}}

\nc{\preprocessingtimehat}{\ensuremath{\hat{t}_p}}
\nc{\inittimehat}{\ensuremath{\hat{t}_i}}
\nc{\delaytimehat}{\ensuremath{\hat{t}_d}}
\nc{\updatetimehat}{\ensuremath{\hat{t}_u}}
\nc{\answertimehat}{\ensuremath{\hat{t}_a}}
\nc{\countingtimehat}{\ensuremath{\hat{t}_c}}
\nc{\testingtimehat}{\ensuremath{\hat{t}_t}}

\nc{\phiBTypical}{\ensuremath{\phi'_{\relS\text{-}\relE\text{-}\relT}}}
\nc{\phiJTypical}{\ensuremath{\phi_{\relS\text{-}\relE\text{-}\relT}}}
\nc{\phiET}{\ensuremath{\phi_{\relE\text{-}\relT}}}

\nc{\restrict}[2]{\ensuremath{{#1}_{|#2}}}
\nc{\extend}[3]{\ensuremath{{#1}\tfrac{#3}{#2}}}
\nc{\valuation}{\ensuremath{\beta}}
\nc{\emptyassign}{\ensuremath{\emptyset}}
\nc{\Assign}[2]{\ensuremath{\frac{#2}{#1}}}

\nc{\vroot}{\ensuremath{\varv_{\textsl{root}}}}

\nc{\pointerfont}[1]{\textit{#1}}

\nc{\varitem}[1]{\ensuremath{v^{#1}}}
\nc{\assitem}[1]{\ensuremath{\assign^{#1}}}
\nc{\constitem}[1]{\ensuremath{a^{#1}}}
\nc{\parentitem}[1]{\ensuremath{\parent^{#1}}}
\nc{\childitem}[2]{\ensuremath{\pointerfont{child}^{#1}_{#2}}}
\nc{\llist}[2]{\ensuremath{\mathcal{L}_{#2}^{#1}}}
\nc{\startlist}{\ensuremath{\mathcal{L}_{\text{\upshape start}}}\xspace}
\nc{\nextlistitem}[1]{\ensuremath{\pointerfont{next-listitem}^{#1}}}
\nc{\prevlistitem}[1]{\ensuremath{\pointerfont{prev-listitem}^{#1}}}
\nc{\countitem}[1]{\ensuremath{C_{\textit{below}}^{#1}}}
\nc{\desc}[1]{\ensuremath{\text{desc}}}

\nc{\Null}{\ensuremath{0}}

\nc{\arrayA}{\arrayfont{A}}
\nc{\arrayB}{\arrayfont{B}}
\nc{\arrayC}{\arrayfont{C}}
\nc{\arrayE}{\arrayfont{E}}

\nc{\ITEMS}{\mathcal{I}}

\nc{\NIL}{\textsc{nil}}

\nc{\TupleSet}{\ensuremath{\mathcal{T}}}
\nc{\ResultSet}{\ensuremath{\mathcal{R}}}
\nc{\SkipArrayNext}[1]{\ensuremath{\mathsf{skip}[#1].\mathsf{next}}}
\nc{\SkipArrayPrev}[1]{\ensuremath{\mathsf{skip}[#1].\mathsf{prev}}}
\nc{\AlgoA}{\ensuremath{\mathbb{A}}}
\nc{\nil}{\texttt{nil}\xspace}
\nc{\SkipStart}{\ensuremath{\mathsf{sk{-}start}}}
\nc{\tup}{\ensuremath{\ov{t}}}
\nc{\tups}{\ensuremath{\ov{s}}}
\nc{\prozvisit}{\ensuremath{\textsc{Visit}}}
\nc{\prozvisitrev}{\ensuremath{\textsc{Visit}^{-1}}}
\nc{\tut}{\ensuremath{t}}
\nc{\enumprev}{\ensuremath{\vartriangleleft}}
\nc{\SkipLast}{\ensuremath{\mathsf{sk{-}last}}}

\nc{\lllist}{\ensuremath{\mathcal{L}}}
\nc{\pllist}{\ensuremath{\mathcal{L}^+}}
\nc{\milist}{\ensuremath{\mathcal{L}^-}}
\nc{\cilist}{\ensuremath{\mathcal{L}^\circ}}
\nc{\numitmpl}{\ensuremath{+{-}\text{on}{-}\text{path}}}
\nc{\numitmmi}{\ensuremath{-{-}\text{on}{-}\text{path}}}
\nc{\numitmci}{\ensuremath{\circ{-}\text{on}{-}\text{path}}}
\nc{\DBnew}{\ensuremath{\DB_{\text{new}}}}
\nc{\DBold}{\ensuremath{\DB_{\text{old}}}}
\nc{\liitmpl}{\ensuremath{\mathcal{L}^{+{-}\text{on}{-}\text{path}}}}
\nc{\liitmmi}{\ensuremath{\mathcal{L}^{-{-}\text{on}{-}\text{path}}}}
\nc{\liitmci}{\ensuremath{\mathcal{L}^{\circ{-}\text{on}{-}\text{path}}}}
\nc{\ITEMSres}[1]{\ensuremath{\ITEMS|_{#1}}}

\nc{\prVisit}{\textsc{Visit}}
\nc{\prVisitRes}{\textsc{VisitRes}}
\nc{\prEnumWithItem}{\textsc{EnumWithItem}}
\nc{\prFindItems}{\textsc{FindItems}}

\nc{\emptytuple}{\ensuremath{()}}
\nc{\emptyword}{\ensuremath{\varepsilon}}

\nc{\proj}{\ensuremath{\pi}}

\nc{\FD}{\ensuremath{\delta_{\textit{fd}}}} 
\nc{\IND}{\ensuremath{\delta_{\textit{ind}}}} 
\nc{\INDtilde}{\ensuremath{\tilde{\delta}_{\textit{ind}}}}
                                
\nc{\SD}{\ensuremath{\delta_{\textit{sd}}}} 
\nc{\CC}{\ensuremath{\delta_{\textit{cc}}}}

\nc{\DEP}{\ensuremath{\delta}} 
\nc{\CONSTR}{\ensuremath{\Sigma}} 

\nc{\qSET}{\ensuremath{q_{\textit{S-E-T}}}}
\nc{\pSET}{\ensuremath{p_{\textit{S-E-T}}}}

\nc{\qET}{\ensuremath{q_{\textit{E-T}}}}

\title{%
 Answering UCQs under updates
 \\ 
 and in the presence of integrity constraints\thanks{Funded by the Deutsche
                  Forschungsgemeinschaft (DFG, German Research Foundation) -- SCHW 837/5-1.}
}

\author{%
  Christoph Berkholz, 
  Jens Keppeler, 
  Nicole Schweikardt 
  \\
  Humboldt-Universität zu Berlin \\
  \texttt{\{berkholz,keppelej,schweika\}@informatik.hu-berlin.de}
}

\begin{document}
\maketitle{}

\begin{abstract}
We investigate the query evaluation problem for fixed queries over
fully dynamic databases where tuples can be inserted or deleted.
The task is to design a dynamic data structure that can immediately
report the new result of a fixed query after every database update.
We consider unions of conjunctive queries (UCQs) and focus on the
query evaluation tasks
\emph{testing} (decide whether an input tuple $\ov{a}$ belongs to the
query result), 
\emph{enumeration} (enumerate, without repetition,
all tuples in the query result), 
and \emph{counting} (output the number of tuples in the
query result).

We identify three increasingly restrictive classes of UCQs which we
call \emph{t-hierarchical}, \emph{q-hierarchical}, and
\emph{exhaustively q-hierarchical} UCQs.
Our main results provide the following dichotomies:
If the query's homomorphic core is t-hierarchical (q-hierarchical,
exhaustively q-hierarchical), then the testing (enumeration, counting)
problem can be solved
with constant update time and constant testing time (delay, counting time). Otherwise, it
cannot be solved with sublinear update time and sublinear testing
time (delay, counting time).$^\ast$

We also study the complexity of query evaluation in the dynamic setting
in the presence of integrity constraints, and we obtain according
dichotomy results
for the special case of small domain constraints (i.e., constraints which state that
all values in a particular column of 
a relation belong to a fixed domain of constant size).

\smallskip
$^\ast)$ To be precise:
our lower bound for the enumeration problem
is obtained only for queries that are self-join free, 
with sublinear we mean $O(n^{1-\smalleps})$ for $\smalleps>0$ and
where $n$ is the size of the active domain of the current database, 
and all our
lower bounds rely on the OV-conjecture 
and/or the OMv-conjecture,
two algorithmic conjectures on the hardness of the Boolean orthogonal vectors
problem and the Boolean online matrix-vector multiplication problem.
\end{abstract}

\section{Introduction}\label{section:introduction}

\emph{Dynamic query evaluation} refers to a setting where a fixed
query $q$ has to be evaluated against a database that is constantly
updated \cite{IdrisUgarteVansummeren.2017}. 
In this paper, we study dynamic query evaluation for
unions of conjunctive queries (UCQs)
on relational databases that may be updated by inserting or deleting tuples.
A dynamic algorithm for evaluating a query $q$
receives an initial database and performs a preprocessing
phase which builds a data structure that contains a suitable
representation of the database and the result of $q$ on this database.
After every database update, the data structure is updated so that it
suitably represents the new database $\DB$ and the result $q(\DB)$ of $q$ on
this database.

To solve the \emph{counting problem}, such an algorithm is required to quickly
report the number $|q(\DB)|$ of tuples in the current query result,
and the \emph{counting time} is the time used to compute this number.
To solve the \emph{testing problem}, the algorithm has to be able to check for an arbitrary
input tuple $\ov{a}$ if $\ov{a}$ belongs to the current query
result, and the \emph{testing time} is the time used to perform this check.
To solve the \emph{enumeration problem}, the algorithm has to enumerate $q(\DB)$
without repetition and with a bounded \emph{delay} between the output tuples.
The \emph{update time} is the time used for updating the data
structure after having received a database update.
We regard the counting (testing, enumeration) problem of a query $q$
to be \emph{tractable under updates} if it can be solved by a dynamic
algorithm with linear preprocessing time, constant update time, and
constant counting time (testing time, delay).

This setting has been studied for conjunctive queries (CQs)
in our previous paper \cite{BKS_enumeration_PODS17}, which identified
a class of CQs called \emph{\qhier} that precisely characterises the
tractability frontier of the counting problem and the enumeration
problem for CQs under updates:
For every q-hierarchical CQ,
the counting problem and the enumeration problem can be solved with 
linear preprocessing time, constant update time, constant counting
time, and constant delay.
And for every CQ that is not equivalent to a \qhier CQ, 
the counting problem (and for the case of self-join free queries, the enumeration problem) 
cannot be solved with
sublinear update time and sublinear counting time (delay),
unless the OMv-conjecture or the OV-conjecture (the OMv-conjecture)
fails.
The latter are well-known algorithmic conjectures on the hardness of the
Boolean online matrix-vector multiplication problem (OMv) and the
Boolean orthogonal vectors problem (OV)
\cite{Henzinger.2015,Abboud.2015},
and ``sublinear'' means $\bigOh(n^{1-\epsilon})$, where $\epsilon>0$
and $n$ is the
size of the active domain of the current database.

\myparagraph{Our contribution}
We identify a new subclass of CQs 
which we call \emph{t-hierarchical}, 
which contains and properly extends the class of q-hierarchical CQs, 
and
which precisely characterises the tractability frontier of the testing
problem for CQs under updates 
(see Theorem~\ref{thm:testing}): 
For every t-hierarchical CQ, the testing problem
can be solved by a dynamic algorithm with linear preprocessing time,
constant update time, and constant testing time.
And for every CQ that is not equivalent to a t-hierarchical CQ,
the testing problem cannot be solved with arbitrary preprocessing
time, sublinear update time, and sublinear testing time, unless the
OMv-conjecture fails.

Furthermore, we transfer the notions of t-hierarchical and
q-hierarchical queries to unions of conjunctive queries (UCQs)
and identify a further class of UCQs 
which we call \emph{exhaustively \qhier}, yielding three increasingly
restricted subclasses of UCQs.
In a nutshell, our main contribution concerning UCQs shows that these
notions precisely characterise the tractability frontiers of the
testing problem, the enumeration problem, and the counting problem for
UCQs under updates
(see the Theorems~\ref{thm:UCQtesting}, \ref{thm:enumUCQ}, \ref{thm:countUCQ}):
For every t-hierarchical (q-hierarchical,
exhaustively q-hierarchical) UCQ, the testing (enumeration, counting)
problem can be solved with linear preprocessing time,
constant update time, and constant testing time (delay, counting
time). 
And for every UCQ that is not equivalent to a t-hierarchical (\qhier,
exhaustively \qhier) UCQ, the testing (enumeration, counting)
problem cannot be solved with sublinear update time and sublinear testing
time (delay, counting time); to be precise, the lower bound for
enumeration is obtained only for self-join free queries, 
the lower bounds for testing and enumeration are conditioned on 
the OMv-conjecture,
and the lower bound for counting is conditioned on the OMv-conjecture and the
OV-conjecture.

Finally, we transfer our results to a scenario where
databases are required to satisfy a set of \emph{small domain constraints}
(i.e., constraints stating that all values which occur in a particular column of
a relation belong to a fixed domain of constant size), leading to a
precise characterisation of the UCQs for which the testing
(enumeration, counting) problem under updates is tractable in this scenario (see
Theorem~\ref{thm:SD-constraints}). 

\myparagraph{Further related work} 
The complexity of evaluating CQs and UCQs in the \emph{static} setting
(i.e., without database updates) 
is well-studied. In particular, there are characterisations of
``tractable'' queries known for Boolean queries
\cite{DBLP:conf/stoc/GroheSS01,Grohe.2007,Marx.2013} as well as for
the task of counting the result tuples
\cite{Dalmau.2004,Chen.2015,Durand.2015,Greco.2014,ChenMengel.2016}. 
In \cite{Bagan.2007},
the fragment of self-join
free CQs that can be enumerated with constant delay after linear
preprocessing time has been identified, but almost nothing
is known about the complexity of the enumeration problem for UCQs on static databases.
Very recent papers also studied the complexity of CQs with respect to
a given set of integrity constraints \cite{DBLP:journals/jacm/GottlobLVV12,DBLP:conf/pods/KhamisNS16,DBLP:conf/pods/BarceloGP16}.
The \emph{dynamic} query evaluation problem has been
considered from different angles, including \emph{descriptive dynamic 
complexity} \cite{Patnaik.1997,Schwentick.2016,Zeume.2014}
and, somewhat closer to what we are aiming for, \emph{incremental
view maintenance} 
\cite{Gupta.1993,
DBLP:journals/ftdb/ChirkovaY12,
DBLP:conf/pods/Koch10,
DBLP:conf/pods/0001LT16,
DBLP:conf/sigmod/NikolicD016}. 
In \cite{IdrisUgarteVansummeren.2017}, the
enumeration and testing problem under updates has been studied 
for q-hierarchical and (more general) acyclic CQs in 
a setting that is very similar to our setting and the setting of \cite{BKS_enumeration_PODS17}; 
the \emph{Dynamic Constant-delay Linear Representations} (DCLR) of \cite{IdrisUgarteVansummeren.2017} are
data structures that use at most linear update time and solve 
the enumeration problem and the testing problem with constant delay and constant testing time.

\myparagraph{Outline}
The rest of the paper is structured as follows.
Section~\ref{section:preliminaries} provides basic notations
concerning databases, queries, and dynamic algorithms for query
evaluation.
Section~\ref{section:CQs} is devoted to CQs and proves
our dichotomy result concerning the {testing} problem for CQs.
Section~\ref{section:UCQs} focuses on UCQs and proves our dichotomies
concerning the {testing}, {enumeration}, and
{counting} problem for UCQs.
Section~\ref{section:QueriesWithIntegrityConstraints} is devoted to
the setting in which integrity constraints may cause a query whose
evaluation under updates is hard in general to be tractable on  
databases that satisfy the constraints.

\section{Preliminaries}\label{section:preliminaries}

\myparagraph{Basic notation}
We write $\NN$ for the set of non-negative integers and let 
$\NNpos\deff\NN\setminus\set{0}$ and $[n]\deff\set{1,\ldots,n}$ for
all $n\in\NNpos$.
By $\potenzmengeof{S}$ we denote the power set of a set $S$.
We write $\vec{v}_i$ to denote the $i$-th component of an
$n$-dimensional vector $\vec{v}$,
and we write $M_{i,j}$ for the entry in row $i$ and column $j$ of a matrix
$M$.
By $\emptytuple$ we denote the empty tuple, i.e., the unique tuple of arity 0.
For an $r$-tuple $t=(t_1,\ldots,t_r)$ and indices $i_1,\ldots,i_m\in\set{1,\ldots,r}$ we write
$\proj_{i_1,\ldots,i_m}(t)$ to denote the projection of $t$ to the components $i_1,\ldots,i_m$, i.e.,
the $m$-tuple $(t_{i_1},\ldots,t_{i_m})$, and in case that $m=1$ we identify the $1$-tuple $(t_{i_1})$ with the
element $t_{i_1}$.
For a set $T$ of $r$-tuples we let $\proj_{i_1,\ldots,i_m}(T)\deff\setc{\proj_{i_1,\ldots,i_m}(t)}{t\in T}$. 

\myparagraph{Databases}
We fix a countably infinite set $\Dom$, the \emph{domain} of potential
database entries. Elements in $\Dom$ are called \emph{constants}.
A \emph{schema} is a finite set $\schema$ of relation symbols, where
each $R\in\schema$ is equipped with a fixed \emph{arity} $\ar(R)\in\NN$ (note that here we explicitly allow 
relation symbols of arity 0). 
Let us fix a schema $\schema=\set{R_1,\ldots,R_s}$, and let
$r_i\deff\ar(R_i)$ for $i\in[s]$.
A \emph{database} $\DB$ of schema $\schema$ ($\schema$-db, for short), 
is of the form $\DB=(R_1^\DB,\ldots,R_s^\DB)$, where $R_i^\DB$ is a finite subset of 
$\Dom^{r_i}$.
The \emph{active domain} $\adom{\DB}$ of $\DB$ is the smallest subset
$A$ of $\Dom$ such that $R_i^\DB\subseteq A^{r_i}$ for all $i\in [s]$.

\myparagraph{Queries}
We fix a countably infinite set $\Var$ of \emph{variables}.
We allow queries to use variables as well as constants.
An \emph{atomic formula} (for short: \emph{atom}) $\qatom$ of schema $\schema$ is of the form 
$R v_1 \cdots v_r$ with $R\in\schema$, $r=\ar(R)$, and
$v_1,\ldots,v_r\in\Var\cup\Dom$.
A \emph{conjunctive formula} of schema $\schema$ is of the form 
\begin{equation}\label{eq:CF}\tag{$*$}
  \exists y_1\,\cdots\,\exists y_\ell \bbody{
   \qatom_1 \uund \cdots \uund \qatom_d
  }
\end{equation}
where $\ell\geq 0$, $d\geq 1$, $\qatom_j$ is an atomic
formula of schema $\schema$ for every $j\in [d]$, and
$y_1,\ldots,y_\ell$ are pairwise distinct elements in $\Var$.
For a conjunctive formula $\phi$ of the form \eqref{eq:CF} we let
$\Vars(\phi)$ (and $\Cons(\phi)$, respectively) be the set of all variables (and constants, respectively) occurring in
$\phi$. 
The set of \emph{free} variables of $\phi$ is
$\free(\phi)\deff \Vars(\phi)\setminus\set{y_1,\ldots,y_\ell}$.
For every variable $x\in\Vars(\phi)$ we let $\atoms_\phi(x)$ (or
$\atoms(x)$, if $\phi$ is clear from the context) be the set of all atoms $\qatom_j$ of $\phi$ such that
$x\in\Vars(\psi_j)$. The formula $\phi$ is called \emph{quantifier-free}
if $\ell=0$, and it is called \emph{self-join free} if no relation
symbol occurs more than once in $\phi$.

For $k\geq 0$, a \emph{$k$-ary conjunctive query} ($k$-ary CQ, for
short) is of the form
\begin{equation}\label{eq:karyCQ}\tag{$**$}
 \{ \ (u_1,\ldots,u_k) \ : \ \phi \ \}
\end{equation}
where $\phi$ is a conjunctive formula of schema $\schema$,
$u_1,\ldots,u_k\in\free(\phi)\cup\Dom$, and $\set{u_1,\ldots,u_k}\cap\Var=\free(\phi)$.
We often write $q_{\phi}(\ov{u})$ for $\ov{u}=(u_1,\ldots,u_k)$ (or $q_\phi$ if $\ov{u}$ is clear from the context) 
to denote such a query.
We let $\Vars(q_{\phi})\deff\Vars(\phi)$,
$\free(q_{\phi})\deff\free(\phi)$, and $\Cons(q_{\phi})\deff\Cons(\phi)\cup(\set{u_1,\ldots,u_k}\cap\Dom)$.
For every $x\in\Vars(q_\phi)$ we let $\atoms_{q_\phi}(x)\deff
\atoms_\phi(x)$, and if $q_\phi$ is clear from the context, we omit
the subscript and simply write
$\atoms(x)$. The CQ $q_\phi$ is called \emph{quantifier-free}
(\emph{self-join free}) if $\phi$ is quantifier-free (self-join free).

The semantics are defined as usual: \
A \emph{valuation} is a mapping $\valuation:\Vars(q_\phi)\cup\Dom\to\Dom$
with $\beta(a)=a$ for every $a\in\Dom$.
A valuation $\beta$ is a \emph{homomorphism} from $q_\phi$ to a $\schema$-db
$\DB$
if for every atom $Rv_1\cdots v_r$ in $q_\phi$ we have 
$\big(\beta(v_1),\ldots,\beta(v_r)\big)\in R^\DB$. 
We sometimes write $\beta:q_\phi\to\DB$ to
indicate that $\beta$ is a homomorphism from $q_\phi$ to $\DB$.
The \emph{query result} $q_{\phi}(\DB)$ of 
a $k$-ary CQ $q_{\phi}(u_1,\ldots,u_k)$ on the $\schema$-db $\DB$ is
defined as the set 
$\{\,\big(\beta(u_1),\ldots,\beta(u_k)\big) \ : \ \text{ $\beta$ is a
  homomorphism from $q_\phi$ to $\DB$}\}$.
If $\ov{x}=(x_1,\ldots,x_k)$ is a list of the free variables of
$\phi$ and $\ov{a}\in\Dom^k$, we sometimes write $\DB\models\phi[\ov{a}]$ to indicate that
there is a homomorphism $\beta:q\to\DB$ with
$\ov{a}=\big(\beta(x_1),\ldots,\beta(x_k)\big)$, for the query $q=q_\phi(x_1,\ldots,x_k)$.

A \emph{$k$-ary union of conjunctive queries} ($k$-ary UCQ)
is of the form
\ $q_{1}(\ov{u}_1) \cup \cdots  \cup  q_{d}(\ov{u}_d)$ \
where $d\geq 1$ and $q_{i}(\ov{u}_i)$ is a $k$-ary CQ of schema $\schema$ for every $i\in [d]$.
The query result of such a $k$-ary UCQ $q$ on a $\schema$-db $\DB$ is 
\,$q(\DB)\deff\bigcup_{i=1}^d q_{i}(\DB)$.

For a $k$-ary query $q$ we write $\Vars(q)$ (and $\Cons(q)$) to denote the set of all 
variables (and constants) that occur in $q$.
Clearly,
$q(\DB)\subseteq(\Adom(\DB)\cup\Cons(q))^k$.

A \emph{Boolean} query is a query of arity $k=0$.
As usual, for Boolean queries $q$ we will write $q(\DB)=\Yes$
instead of $q(\DB)\neq\emptyset$, and $q(\DB)=\No$ instead of 
$q(\DB)=\emptyset$.
Two $k$-ary queries $q$ and $q'$ are 
\emph{equivalent} ($q\equiv q'$, for short) 
if $q(\DB)=q'(\DB)$ for every $\schema$-db $\DB$.

\myparagraph{Homomorphisms}
We use standard notation concerning homomorphisms
(cf., e.g.\ \cite{AHV-Book}).
The notion of a homomorphism $\beta:q\to\DB$ from a CQ $q$ to a
database $\DB$ has already been defined above.
A homomorphism $g:\DB\to q$ from a database $\DB$ to a CQ $q$ is a
mapping from $\Adom(\DB)\to\Vars(q)\cup\Cons(q)$ such that whenever
$(a_1,\ldots,a_r)$ is a tuple in some relation $R^\DB$ of $\DB$, then 
$Rg(a_1)\cdots g(a_r)$ is an atom of $q$.

Let  $q(u_1,\ldots,u_k)$ and $q'(v_1,\ldots,v_k)$ be two
$k$-ary CQs. A \emph{homomorphism} from $q$ to $q'$ is a
mapping $h\colon \Vars(q)\cup\Dom \to \Vars(q')\cup\Dom$ with
$h(a)=a$ for all $a\in\Dom$ and $h(u_i)=v_i$ for all $i\in [k]$ such
that for every atom $Rw_1\cdots w_r$ in $q$ there is an atom
$Rh(w_1)\cdots h(w_r)$ in $q'$. We sometimes write $h:q\to q'$ to
indicate that $h$ is a homomorphism from $q$ to $q'$.
Note that by \cite{DBLP:conf/stoc/ChandraM77}
there is a homomorphism from $q$
to $q'$ if and only if for every database $\DB$ it holds that
$q(\DB)\supseteq q'(\DB)$.
A CQ $q$ is a \emph{homomorphic core} if there is no homomorphism
from $q$ into a proper subquery of $q$. 
Here, a \emph{subquery} of a CQ $q_\phi(\ov{u})$ where $\phi$ is of
the form \eqref{eq:CF} is a CQ $q_{\phi'}(\ov{u})$ where
$\phi'$ is of the form $\exists y_{i_1}\cdots\exists y_{i_m}\;
(\psi_{j_1} \uund \cdots\uund \psi_{j_n})$ with
$i_1,\ldots,i_m\in[\ell]$, $j_1,\ldots,j_n\in[d]$,
and $\free(\phi')=\free(\phi)$.

We say that a UCQ is
a \emph{homomorphic core}, if every CQ in the union is a homomorphic core and there is no
homomorphism between two distinct CQs.
It is well-known that every CQ and every UCQ is equivalent to a unique
(up to renaming of variables) homomorphic core, which is therefore
called \emph{the core of} the query (cf., e.g., \cite{AHV-Book}). 

\myparagraph{Sizes and Cardinalities}
The \emph{size} $\size{\schema}$ of a schema $\schema$ is 
$|\schema|+\sum_{R\in\schema}\ar(R)$.
The size $\size{q}$ of a query $q$ of schema $\schema$ is 
the length of $q$ when viewed as a word over the alphabet 
$\schema\cup\Var\cup\Dom\cup\set{\,\und\,,\exists\,,(\,,)\,,\{\,,\}\,,:\,,\cup\,}\cup\set{\,,}$.
For a $k$-ary query $q$ and a $\sigma$-db $\DB$, the 
\emph{cardinality of the query result} is the number $|q(\DB)|$ of
tuples in $q(\DB)$.
The \emph{cardinality} $\card{\DB}$ of a $\schema$-db $\DB$ is defined
as the number of tuples stored in $\DB$, i.e.,
$\card{\DB}\deff\sum_{R\in\schema} |R^{\DB}|$. 
The \emph{size} $\size{\DB}$ of $\DB$ is defined as
$\size{\schema}+|\Adom(\DB)|+\sum_{R\in\schema} \ar(R){\cdot}
|R^D|$ and corresponds to the size of a reasonable encoding of $\DB$.

\medskip

The following notions concerning updates, dynamic algorithms for query
evaluation, and algorithmic conjectures are taken almost verbatim from
\cite{BKS_enumeration_PODS17}.

\myparagraph{Updates}
We allow to update a given database of schema $\schema$ by inserting or deleting
tuples as follows. An \emph{insertion} command is of the form
 \Insert\ $R(a_1,\ldots,a_r)$
for $R\in\schema$, $r=\ar(R)$, and $a_1,\ldots,a_r\in \Dom$. When
applied to a $\schema$-db $\DB$, it results in the updated $\schema$-db
$\DB'$ with $R^{\DB'}\deff R^{\DB}\cup\set{(a_1,\ldots,a_r)}$ and
$S^{\DB'}\deff S^{\DB}$ for all $S\in\schema\setminus\set{R}$.
A \emph{deletion} command is of the form
 \Delete\ $R(a_1,\ldots,a_r)$
for $R\in\schema$, $r=\ar(R)$, and $a_1,\ldots,a_r\in \Dom$. When
applied to a $\schema$-db $\DB$, it results in the updated $\schema$-db
$\DB'$ with $R^{\DB'}\deff R^{\DB}\setminus\set{(a_1,\ldots,a_r)}$ and
$S^{\DB'}\deff S^{\DB}$ for all $S\in\schema\setminus\set{R}$.
Note that both types of commands may change the database's active domain.

\myparagraph{Dynamic algorithms for query evaluation}
Following \cite{Cormen.2009}, we use Random Access Machines (RAMs)
with $\bigoh(\log n)$ word-size and a uniform cost 
measure to analyse our algorithms.
We will assume that the RAM's memory is initialised to $\Null$. In
particular, if an algorithm uses an array, we will assume
that all array entries are initialised to $\Null$, and this initialisation
comes at no cost (in real-world computers this can be achieved by using the
\emph{lazy array initialisation technique}, cf.\ e.g.\ \cite{MoretShapiro}). 
A further assumption is that for every fixed
dimension $k\in\NNpos$ we have available an unbounded number of
$k$-ary arrays $\arrayA$ such that for given $(n_1,\ldots,n_k)\in\NN^k$
the entry $\arrayA[n_1,\ldots,n_k]$ 
at position $(n_1,\ldots,n_k)$ can be accessed in constant
time.\footnote{While this can be accomplished easily in the RAM-model,  
for an implementation on real-world
computers one would probably have to resort to replacing our use of
arrays by using suitably designed hash functions.}
For our purposes it will be convenient to assume that $\Dom=\NNpos$.

Our algorithms will take as input 
a $k$-ary query $q$
and a $\schema$-db $\DBstart$.
For all query evaluation problems considered in this paper, we aim at
routines $\PREPROCESS$ and $\UPDATE$ which achieve the following.
Upon input of $q$
and $\DBstart$, the $\PREPROCESS$ routine builds a data
structure $\DS$ which represents $\DBstart$ (and which is designed in
such a way that it supports the evaluation of $q$ on $\DBstart$).
Upon input of a command $\Update\ R(a_1,\ldots,a_r)$ (with
$\Update\in\set{\Insert,\Delete}$), 
calling $\UPDATE$  modifies the data structure $\DS$ such that it
represents the updated database $\DB$.
The \emph{preprocessing time} $\preprocessingtime$ is the
time used for performing $\PREPROCESS$.
The \emph{update time} $\updatetime$ is the time used for performing
an $\UPDATE$, and in this paper we aim at algorithms where $\updatetime$ is independent of the size
of the current database $\DB$.
By $\INIT$ we denote the particular case of the routine $\PREPROCESS$
upon input of a query $q$
and the \emph{empty} database
$\DBempty$, where $R^{\DBempty}=\emptyset$ for all $R\in\schema$.
The \emph{initialisation time} $\inittime$
is the time used for performing $\INIT$.
In all algorithms presented in this paper, the $\PREPROCESS$ routine
for input of $q$ and $\DBstart$ 
will carry out the $\INIT$ routine for $q$ 
and then perform a sequence of $\card{\DBstart}$ update operations to
insert all the tuples of $\DBstart$ into the data structure.
Consequently, $\preprocessingtime = \inittime +  \card{\DBstart}\cdot\updatetime$.

In the following, $\DB$ will always denote the database that is
currently represented by the data structure $\DS$.
To solve the \emph{enumeration problem under updates}, 
apart from the routines $\PREPROCESS$ and $\UPDATE$,
we aim at a routine $\ENUMERATE$ such that
calling $\ENUMERATE$ invokes an enumeration of all tuples,
\emph{without repetition}, that belong to the query result $q(\DB)$.
The \emph{delay} $\delaytime$ is the maximum time
used during a call of $\ENUMERATE$
\begin{itemize}
\item until the output of the first tuple (or the end-of-enumeration
  message $\EOE$, 
  if $q(\DB)=\emptyset$),
\item between the output of two consecutive tuples, and 
\item between the output of the last tuple and the end-of-enumeration
  message $\EOE$.
\end{itemize}

To \emph{test} if a given tuple belongs to the query result,
instead of $\ENUMERATE$ we aim at a routine $\TEST$ which
upon input of a tuple $\ov{a}\in\Dom^k$ checks whether $\ov{a}\in
q(\DB)$.
The \emph{testing time} $\testingtime$ is the time used for
performing a $\TEST$.
To solve the \emph{counting problem under updates}, 
we aim at a routine $\COUNT$ which outputs the cardinality
$|q(\DB)|$ of the query result.
The \emph{counting time} $\countingtime$ is the time used for
performing a $\COUNT$.
To \emph{answer} a \emph{Boolean} query under updates,
we aim at a routine $\ANSWER$
that produces the answer $\Yes$ or $\No$ of $q$ on $\DB$.
The \emph{answer time} $\answertime$ is the time used for
performing $\ANSWER$.
Whenever speaking of a \emph{dynamic algorithm}, we mean an algorithm
that has routines $\PREPROCESS$ and $\UPDATE$ and, depending on the
problem at hand, at least one of the routines 
$\ANSWER$,
$\TEST$, 
$\COUNT$,
and $\ENUMERATE$.

Throughout the paper, we often adopt the view of \emph{data complexity} 
and suppress factors that may depend on the query $q$ 
but not on the database $\DB$. 
E.g., ``linear preprocessing time'' means 
$\preprocessingtime\leq f(q)\cdot\size{\DBstart}$ and 
``constant update time'' means  $\updatetime\leq f(q)$, for a
function $f$ with codomain $\NN$. 
When writing $\poly(n)$ we mean $n^{\bigOh(1)}$, and
for a query $q$ we often write $\poly(q)$ instead of $\poly(\size{q})$.

\myparagraph{Algorithmic conjectures}
Similarly as in \cite{BKS_enumeration_PODS17} we obtain hardness
results that are conditioned on algorithmic conjectures concerning the
hardness of the following problems. These problems deal with
\emph{Boolean} matrices and vectors, i.e., matrices and vectors over
$\set{0,1}$, and all the arithmetic is done over the Boolean semiring,
where multiplication means conjunction and addition means disjunction.

The \emph{orthogonal vectors
  problem} (OV-problem) is the following decision problem. Given two sets $U$
and $V$ of $n$ Boolean vectors of dimension $d$, decide whether there
are vectors $\vec{u}\in U$ and $\vec{v}\in V$ such that
$\vec{u}\trans\vec{v}=0$.
The \emph{OV-conjecture} states that there is no $\epsilon>0$ such that
the OV-problem for $d=\aufgerundet{\log^2 n}$ can be solved in time
$\bigOh(n^{2-\epsilon})$, see \cite{Abboud.2015}.

The \emph{online matrix-vector multiplication problem} (OMv-problem) is the
following algorithmic task. At first, the algorithm gets a Boolean
$n\times n$ matrix $M$ and is allowed to do some
preprocessing. Afterwards, the algorithm receives $n$ vectors
$\vec{v}^{\,1},\ldots,\vec{v}^{\,n}$ one by one and has to output
$M\vec{v}^{\,t}$ before it has access to $\vec{v}^{\,t+1}$ (for each
$t<n$).
The running time is the overall time the algorithm needs to produce
the output $M\vec{v}^{\,1},\ldots,M\vec{v}^{\,n}$.
The \emph{OMv-conjecture} \cite{Henzinger.2015} states that there is no
$\epsilon>0$ such that the OMv-problem can be solved in time $\bigOh(n^{3-\epsilon})$. 

A related problem is the \emph{OuMv-problem} where the
algorithm, again, is given a Boolean $n\times n$ matrix $M$ and is
allowed to do some preprocessing. Afterwards, the algorithm receives a
sequence of pairs of $n$-dimensional Boolean vectors
$\vec{u}^{\,t},\vec{v}^{\,t}$ for each $t\in[n]$, and the task is to
compute $(\vec{u}^{\,t})\trans M \vec{v}^{\,t}$ before accessing
$\vec{u}^{\,t+1},\vec{v}^{\,t+1}$.
The \emph{OuMv-conjecture} states that there is no $\epsilon>0$ such
that the OuMv-problem can be solved in time
$\bigOh(n^{3-\epsilon})$. 
It was shown in \cite{Henzinger.2015} that the OuMv-conjecture is
equivalent to the OMv-conjecture, i.e., 
the OuMv-conjecture fails if, and only if, the OMv-conjecture fails.

\section{Conjunctive queries}\label{section:CQs}

This section's aim is twofold: Firstly, we observe that the notions and
results of \cite{BKS_enumeration_PODS17} generalise to  
CQs with constants in a straightforward way.
Secondly, we identify a new subclass of CQs which precisely
characterises the CQs for which \emph{testing} can be done efficiently
under updates.

The definition of \emph{q-hierarchical} CQs can be taken verbatim from
\cite{BKS_enumeration_PODS17}:
\begin{definition}\label{def:qhierarchical}
A CQ $q$ is \emph{q-hierarchical} if for any two variables
$x,y\in\Vars(q)$ we have
\begin{enumerate}[(i)]
\item\label{eq:q-hier-cond-i}
 $\atoms(x)\subseteq\atoms(y)$ 
 \ or \
 $\atoms(y)\subseteq\atoms(x)$
 \ or \ 
 $\atoms(x)\cap\atoms(y)=\emptyset$,
 \ and
\item\label{eq:q-hier-cond-ii}
 if $\atoms(x)\varsubsetneq\atoms(y)$ and $x\in\free(q)$, then $y\in\free(q)$.
\end{enumerate}
\end{definition}

Obviously, it can be checked in time $\poly(q)$ whether a given CQ $q$
is q-hierarchical.
It is straightforward to see that if a CQ is q-hierarchical, then so
is its homomorphic core.
Using the main results of \cite{BKS_enumeration_PODS17}, it is not
difficult to show the following.

\begin{theorem}\label{thm:CQs}\label{thm:CQupper}\label{thm:CQlower}
\begin{enumerate}[(a)]
\item\label{item:thm:CQs:upper}
There is a dynamic algorithm that receives a q-hierarchical $k$-ary CQ $q$ and a $\schema$-db $\DBstart$, and computes
within $\preprocessingtime =\poly({q})\cdot\bigOh(\size{\DBstart})$ preprocessing time a data structure that can be
updated in time $\updatetime = \poly({q})$ and allows to
\begin{enumerate}[(i)]
\item\label{item:thm:CQs:count}
 compute the cardinality $|q(\DB)|$ in time $\countingtime =\bigOh(1)$,
\item\label{item:thm:CQs:enumerate} 
 enumerate $q(\DB)$ with delay $\delaytime = \poly({q})$,
\item\label{item:thm:CQs:test}
 test for an input tuple $\ov{a}\in\Dom^k$ if $\ov{a}\in q(\DB)$ within time
 $\testingtime = \poly({q})$,
\item\label{item:thm:CQs:next}
 and when given a tuple $\ov{a}\in q(\DB)$, the tuple $\ov{a}'$ (or
 the message $\EOE$) that the enumeration procedure of  
 \eqref{item:thm:CQs:enumerate} would output directly after
 having output $\ov{a}$, can be computed within time 
 $\poly({q})$.
\end{enumerate}
\item\label{item:thm:CQs:lower}
Let $\epsilon>0$ and let $q$ be a CQ whose homomorphic core is not
q-hierarchical 
(note that this is the case if, and only if, $q$ is not equivalent to a q-hierarchical CQ).
\begin{enumerate}[(i)]
\item\label{item:thm:CQlower:answering}
  If $q$ is Boolean, then there is no dynamic algorithm  with arbitrary
  preprocessing time and $\updatetime=\bigoh(\actdomsize^{1-\smalleps})$
  update time that
  answers 
  $\eval{q}{\DB}$ in time
$\answertime=\bigoh(\actdomsize^{2-\smalleps})$, unless the \OMvcon
  fails.
\item\label{item:thm:CQlower:counting}
  There is no dynamic algorithm  with arbitrary
  preprocessing time and $\updatetime=\bigoh(\actdomsize^{1-\smalleps})$
  update time that
  computes the cardinality $|q(\DB)|$ in time
$\countingtime=\bigoh(\actdomsize^{1-\smalleps})$, unless the \OMvcon
or the \OVcon fails.
\item\label{item:thm:CQs:enumeration} 
  If $q$ is self-join free, then there is no dynamic algorithm  with 
  arbitrary
  preprocessing time and  $\updatetime=\bigoh(\actdomsize^{1-\smalleps})$
  update time that
  enumerates
  $\eval{q}{\DB}$
   with delay $\delaytime = \bigoh(\actdomsize^{1-\smalleps})$, unless the \OMvcon
  fails.
\end{enumerate}
All lower bounds remain true, if we restrict ourselves to the class of
databases that map homomorphically into $q$.
\end{enumerate}
\end{theorem}

\begin{proof}
From \cite{BKS_enumeration_PODS17} we already know that the theorem's
statements \eqref{item:thm:CQs:count} and
\eqref{item:thm:CQs:enumerate} 
and
\eqref{item:thm:CQlower:answering}--\eqref{item:thm:CQs:enumeration}
are true for all CQs $q$ with $\Cons(q)=\emptyset$,
and a close look at the dynamic algorithm provided in
\cite{BKS_enumeration_PODS17} shows that also the statements 
\eqref{item:thm:CQs:test} and \eqref{item:thm:CQs:next} are true for
all CQs $q$ with $\Cons(q)=\emptyset$.
Furthermore, a close inspection of the proofs provided in
\cite{BKS_enumeration_PODS17} for the statements 
\eqref{item:thm:CQlower:answering}--\eqref{item:thm:CQs:enumeration}
for constant-free CQs $q$ shows that with only very minor
modifications these proofs carry over to the case of CQs $q$ with
$\Cons(q)\neq \emptyset$.

All that remains to be done is to transfer the results 
\eqref{item:thm:CQs:count}--\eqref{item:thm:CQs:next}
from constant-free CQs to CQs $q$ with $\Cons(q)\neq \emptyset$.
To establish this, 
let us consider an arbitrary CQ $q$ of schema $\schema$ with
$\Cons(q)\neq\emptyset$.
Without loss of generality we can assume that 
\begin{equation}
 q \ = \ \ 
 \{ \ (x_1,\ldots,x_k, b_1, \ldots, b_\ell) \ : \ \phi \ \}
\end{equation}
where $\phi$ is a conjunctive formula of schema $\schema$,  
$\free(\phi) = \set{x_1,\ldots,x_k}$, and
$b_1,\ldots,b_\ell~\in~\Dom$. Let $\ov{x}\deff (x_1,\ldots,x_k)$ and
$\ov{b}\deff (b_1,\ldots,b_\ell)$.

In the following, we construct a new schema $\hat{\schema}$ (that
depends on $q$) and a
constant-free CQ $\hat{q}$ of schema $\hat{\schema}$ and of size
$\poly(\size{q})$ such that the
following is true:
\begin{enumerate}
\item
$\hat{q}$ is q-hierarchical \ $\iff$ \ $q$ is q-hierarchical. 
\item
A dynamic algorithm for evaluating $\hat{q}$ on
$\hat{\schema}$-dbs with initialisation time 
$\inittimehat$, update time $\updatetimehat$, counting time
$\countingtimehat$ (delay $\delaytimehat$, testing time $\testingtimehat$)
can be used to obtain a dynamic algorithm for
evaluating $q$ on $\schema$-dbs with 
initialisation time 
$\inittimehat$, update time $\updatetimehat{\cdot}\poly(\size{q})$, counting time
$\countingtimehat$ (delay $\bigOh(\delaytimehat)+\poly(\size{q})$,
testing time $\bigOh(\testingtimehat)+\poly(\size{q})$). 
\end{enumerate}

For each atom $\qatom$ of $q$ we introduce a new relation symbol
$R_\qatom$ of arity $|\Vars(\qatom)|$, and we let
$\hat{\schema}\deff \set{R_{\qatom}\ :\ \qatom $ is an atom in $\phi}$.
For each atom $\qatom$ of $q$ let us fix a tuple
$\tupleVariables^{\qatom}=(v_1,\ldots,v_m)$ of pairwise
distinct variables such that  $\Vars(\qatom)= \set{v_1,\ldots,v_m}$.
The CQ $\hat{q}$ is defined as 
\[
 \hat{q} \ \deff \ \  \{ \ (x_1,\ldots,x_k) \ : \ \hat{\phi} \ \}\,,
\]
where the conjunctive formula $\hat{\phi}$ is obtained from $\phi$ by
replacing every atom $\qatom$ with the atom $R_\qatom(\tupleVariables^\qatom)$. 
Obviously, $\hat{q}$ is a CQ of schema $\hat{\schema}$, 
$\Cons(\hat{q})=\emptyset$, $\free(\hat{q})=\free(q)$, and
$\Vars(\hat{q})=\Vars(q)$.
Furthermore, for every variable $y\in\Vars(q)$ we
have
$\atoms_{\hat{q}}(y) = \setc{R_\qatom}{\qatom\in\atoms_q(y)}$ (and,
equivalently,
$\atoms_{q}(y) = \setc{\qatom}{R_\qatom\in\atoms_{\hat{q}}(y)}$).
Therefore, $\hat{q}$ is q-hierarchical if and only if $q$ is q-hierarchical.

With every $\schema$-db $\DB$ we associate a $\hat{\schema}$-db
$\hat{\DB}$ as follows:
Consider an atom $\qatom$ of $q$ and let $\qatom$ be of the form
$Sw_1\cdots w_s$. Thus, $\set{w_1,\ldots,w_s}\cap\Var=\Vars(\qatom)
=\set{v_1,\ldots,v_m}$ for
$(v_1,\ldots,v_m)\deff\tupleVariables^{\qatom}$. Then, the relation symbol
$R_\psi$ is interpreted in $\hat{\DB}$ by the relation
\[
  (R_\qatom)^{\hat\DB} \ \isdef \ \ 
  \setc{\;
    \big(\valuation(v_1),\ldots,\valuation(v_{m})\big)}{
    \text{$\beta$ is a valuation with } \big(\valuation(w_1),\ldots,\valuation(w_s)\big)\in S^\DB
  \;}\,.
\]
It is straightforward to verify that for every $\schema$-db $\DB$ we
have
\begin{equation}\label {eq:thm:CQs:proof-upper}
  q(\DB) \quad = \quad
  \setc{\;(\ov{a},\ov{b})\;}{\;\ov{a}\ \in\ \hat{q}(\hat{\DB})\;}\,.
\end{equation}

Now, assume we have available a 
dynamic algorithm $\mathcal{A}$ for evaluating $\hat{q}$ on
$\hat{\schema}$-dbs with preprocessing time 
$\preprocessingtimehat$, update time $\updatetimehat$, counting time
$\countingtimehat$ (delay $\delaytimehat$, testing time
$\testingtimehat$).
We can use this algorithm to obtain a dynamic algorithm
$\mathcal{B}$ for
evaluating $q$ on $\schema$-dbs as follows.

The $\INIT$ routine of $\mathcal{B}$ performs the 
$\INIT$ routine of $\mathcal{A}$.
The $\UPDATE$ routine of $\mathcal{B}$ proceeds as follows.
Upon input of an update command of the form 
$\Update\;S(c_1,\ldots,c_s)$ for some $S\in\schema$,
we consider all atoms $\qatom$ of $q$ of the form
$Sw_1\cdots w_s$.
For each such atom we check if
\begin{itemize}
\item
for all $i\in[s]$ with $w_i\in\Dom$ we have $w_i=c_i$, and
\item
for all $i,j\in[s]$ with $w_i=w_j$ we have $c_i=c_j$.
\end{itemize}
If this is true, we carry out the $\UPDATE$ routine of $\mathcal{A}$ for
the command $\Update\,R_\qatom(c_{j_1},\ldots,c_{j_m})$, where
$(w_{j_1},\ldots,w_{j_m})=(v_1,\ldots,v_m)= \tupleVariables^{\qatom}$.
Thus, one call of the $\UPDATE$ routine of $\mathcal{A}$ performs
$\poly(\size{q})$ calls of the $\UPDATE$ routine of
$\mathcal{B}$. This takes time $\updatetimehat{\cdot}\poly(\size{q})$
and ensures that afterwards, the data structure of $\mathcal{B}$ has
stored the information concerning the $\hat{\schema}$-db $\hat{\DB}$
associated with the updated $\schema$-db $\DB$.

The $\COUNT$ routine of $\mathcal{B}$ simply calls the $\COUNT$
routine of $\mathcal{A}$, and we know that the result is correct since
$|q(\DB)|=|\hat{q}(\hat{\DB})|$ due to \eqref{eq:thm:CQs:proof-upper}.
For the same reason, the $\ENUMERATE$ routine of $\mathcal{B}$ can
call the $\ENUMERATE$ routine of $\mathcal{A}$ and output the tuple
$(\ov{a},\ov{b})$ for each output tuple $\ov{a}$ of $\mathcal{A}$.
The $\TEST$ routine of $\mathcal{B}$ upon input of a tuple
$(c_1,\ldots,c_{k+\ell})\in\Dom^{k+\ell}$ outputs $\Yes$ if
$(c_{k+1},\ldots,c_{k+\ell})=\ov{b}$ and the $\TEST$ routine of
$\mathcal{A}$ returns $\Yes$ upon input of the tuple $(c_1,\ldots,c_k)$.
For statement  \eqref{item:thm:CQs:next} of Theorem~\ref{thm:CQs},
when given a tuple $(\ov{a},\ov{b})\in q(\DB)$ we know that 
$\ov{a}\in\hat{q}(\hat{\DB})$. Thus, we can use $\ov{a}$ and call the according routine of 
$\mathcal{A}$ for \eqref{item:thm:CQs:next} and obtain a tuple
$\ov{a}'\in \hat{q}(\hat{\DB})$ (or the message $\EOE$) and know that $(\ov{a}',\ov{b})$ is
the next tuple that the $\ENUMERATE$ routine of $\mathcal{B}$ will output
after having output the tuple $(\ov{a},\ov{b})$ (or that there is no
such tuple).

Note that this suffices to transfer the statements 
\eqref{item:thm:CQs:count}--\eqref{item:thm:CQs:next} from a
q-hierarchical CQ $\hat{q}$ with $\Cons(\hat{q})=\emptyset$ to 
the q-hierarchical CQ $q$ with $\Cons(q)\neq\emptyset$.
This completes the proof of Theorem~\ref{thm:CQs}.
\end{proof}

Note that neither the results of \cite{BKS_enumeration_PODS17} nor 
Theorem~\ref{thm:CQupper} provide a precise
characterisation of the CQs for which \emph{testing} can be done efficiently
under updates.
Of course, according to Theorem~\ref{thm:CQupper}\,\eqref{item:thm:CQs:test}, the testing
problem can be solved with constant update time and constant testing
time for every
\emph{q-hierarchical} CQ. But the same holds true, for example, for
the non-q-hierarchical CQ 
\ $
 \pSET  \deff  \{\,(x,y) \, : \, Sx\uund Exy\uund Ty\,\}\,
 $.
The according dynamic algorithm simply uses 1-dimensional arrays
$\arrayfont{A}_S$ and
$\arrayfont{A}_T$ and a 2-dimensional array $\arrayfont{A}_E$ and that
for all $a,b\in\Dom$ we have
$\arrayfont{A}_S[a]=1$ if $a\in S^\DB$, and $\arrayfont{A}_S[a]=0$ otherwise,
$\arrayfont{A}_T[a]=1$ if $a\in T^\DB$, and $\arrayfont{A}_T[a]=0$
   otherwise, and
$\arrayfont{A}_E[a,b]=1$ if $(a,b)\in E^\DB$, and $\arrayfont{A}_E[a,b]=0$ otherwise.
When given an update command, the arrays can be updated within
constant time. And when given a tuple $(a,b)\in\Dom^2$, the $\TEST$
routine simply looks up the array entries $\arrayfont{A}_S[a]$,
$\arrayfont{A}_E[a,b]$, $\arrayfont{A}_T[b]$ and returns the correct
query result accordingly.
To characterise the conjunctive queries for which testing can be done
efficiently under updates, we introduce the following notion of
\emph{t-hierarchical} CQs.

\begin{definition}\label{def:thierarchical}
A CQ $q$ is \emph{t-hierarchical} if the following is satisfied:
\begin{enumerate}[(i)]
\item\label{eq:t-hier-cond-i}
 for all $x,y\in\Vars(q)\setminus\free(q)$, we have
 
 $\atoms(x)\subseteq\atoms(y)$ 
 \ or \
 $\atoms(y)\subseteq\atoms(x)$
 \ or \ 
 $\atoms(x)\cap\atoms(y)=\emptyset$,
 \ and
\item\label{eq:t-hier-cond-ii}
 for all $x\in\free(q)$ and all $y\in\Vars(q)\setminus\free(q)$, we
 have

 $\atoms(x)\cap\atoms(y)=\emptyset$ \ or \
 $\atoms(y)\subseteq \atoms(x)$.
\end{enumerate}
\end{definition}
 
Obviously, it can be checked in time $\poly(q)$ whether a given CQ $q$
is t-hierarchical.
Note that every q-hierarchical CQ is t-hierarchical, and
a \emph{Boolean} query is t-hierarchical if and only if it is q-hierarchical.
The queries $\pSET$ and 
\ $
  p_{\textit{E-E-R}} \deff 
  \{\, 
  (x,y) \, : \, 
  \exists v_1\exists v_2\exists v_3 \, (\,
    Exv_1 \uund Eyv_2 \uund Rxyv_3
   \,)
  \, \}
$ \ 
are examples for queries that are t-hierarchical but not
q-hierarchical. 
It is straightforward to verify that if a CQ is t-hierarchical, then so
is its homomorphic core.
This section's main result shows that the
\emph{t-hierarchical} CQs 
precisely characterise the CQs for which the \emph{testing} problem
can be solved efficiently under updates:

\begin{theorem}\label{thm:testing}
\begin{enumerate}[(a)]
\item\label{item:thm:testing:upperbound}
There is a dynamic algorithm that receives a t-hierarchical 
$k$-ary CQ $q$ and a $\schema$-db $\DBstart$, and computes within 
$\preprocessingtime =\poly({q})\cdot\bigOh(\size{\DBstart})$ preprocessing time a data structure that can be
updated in time $\updatetime = \poly({q})$ and allows to 
test for an input tuple $\ov{a}\in\Dom^k$ if $\ov{a}\in q(\DB)$ within
time $\testingtime=\poly({q})$.

\item\label{item:thm:testing:lowerbound}
Let $\epsilon >0$ and let $q$ be a $k$-ary CQ whose homomorphic core
is not t-hierarchical (note that this is the case if, and only
if, $q$ is not equivalent to a t-hierarchical CQ).
There is no dynamic algorithm with arbitrary preprocessing time and $\updatetime =\bigOh(n^{1-\epsilon})$
update time that can test for any input tuple $\ov{a}\in\Dom^k$ if
$\ov{a}\in q(\DB)$ within testing time
$\testingtime=\bigOh(n^{1-\epsilon})$, unless the OMv-conjecture
fails.
The lower bound remains true, if we restrict ourselves to the class of
databases that map homomorphically into $q$.
\end{enumerate}
\end{theorem} 

\begin{proof}
To avoid notational clutter, and without loss of generality, we
restrict attention
to queries $q_\phi(u_1,\ldots,u_k)$ where
$(u_1,\ldots,u_k)$ is of the form $(z_1,\ldots,z_k)$ for pairwise
distinct variables $z_1,\ldots,z_k$.

For the proof of \eqref{item:thm:testing:upperbound}, we combine the
array construction described above for the example query $\pSET$ with
the dynamic algorithm provided by
Theorem~\ref{thm:CQupper}\,\eqref{item:thm:CQs:upper} 
and the
following Lemma~\ref{lemma:t-hierCQ-Decomp}.
To formulate the lemma, we need the following notation.
A $k$-ary \emph{generalised CQ} is of the form
\ $
  \{\,
  (z_1,\ldots,z_k) \, : \, 
  \phi_1 \uund \cdots \uund \phi_m
  \, \} 
$ \
where $k\geq 0$, $z_1,\ldots,z_k$ are pairwise distinct variables, $m\geq 1$, $\phi_j$ is a
conjunctive formula for each $j\in[m]$, 
$\free(\phi_1)\cup\cdots\cup\free(\phi_m) =
\set{z_1,\ldots,z_k}$, and the quantified
variables of $\phi_j$ and $\phi_{j'}$ are pairwise disjoint for all
$j,j'\in [m]$ with $j\neq j'$ and disjoint from
$\set{z_1,\ldots,z_k}$.
For each $j\in[m]$ let $\ov{z}^{(j)}$ be the
sublist of $\ov{z}\deff (z_1,\ldots,z_k)$ that only contains the variables in $\free(\phi_j)$.
I.e., $\ov{z}^{(j)}$ is obtained from $\ov{z}$ by deleting all variables
that do not belong to $\free(\phi_j)$. Accordingly, for a tuple
$\ov{a}=(a_1,\ldots,a_k)\in\Dom^k$ by $\ov{a}^{(j)}$ we denote the tuple
that contains exactly those $a_i$ where $z_i$ belongs to $\ov{z}^{(j)}$.
The query result of $q$ on a $\schema$-db $\DB$ is the set 
\[
  q(\DB) \ \deff \ \ 
  \setc{\ \ov{a}\in\Dom^k}{\DB\models\phi_j[\ov{a}^{(j)}] \text{ \ for
      each $j\in[m]$} \ }\,,
\]
where $\DB\models\phi_j[\ov{a}^{(j)}]$ means that
there is a homomorphism
$\beta_j:q_j\to\DB$ for the query
$q_j\deff\{\,\ov{z}^{(j)}\;:\;\phi_j\,\}$, with $\beta_j(z_i)=a_i$
for every $i$ with $z_i\in\free(\phi_j)$. 
For example, 
\ $
  p'_{\textit{E-E-R}} \deff  
  \{\; 
  (x,y) \; : \; 
  \exists v_1\, Exv_1
  \ \und \
  \exists v_2\, Eyv_2
  \ \und \
  \exists v_3\, Rxyv_3
  \; \}
$ \
is a generalised CQ that is equivalent to the CQ $p_{\textit{E-E-R}}$.

\begin{lemma}\label{lemma:t-hierCQ-Decomp}
Every t-hierarchical CQ $q_\phi(z_1,\ldots,z_k)$ is equivalent to a
generalised CQ 
\ $
  q'  =  \{\,
  (z_1,\ldots,z_k) \, : \,
  \phi_1 \uund \cdots \uund \phi_m
  \, \} 
$ \
such that for each $j\in[m]$ the CQ
\ $
  q_j \deff \{\;
   \ov{z}^{(j)} \, : \, \phi_j
  \; \}$ \
is q-hierarchical or quantifier-free.
Furthermore, there is an algorithm which decides in time
$\poly({q_\phi})$ whether $q_\phi$ is \mbox{t-hierarchical}, and if so,
outputs an according $q'$.
\end{lemma}
\begin{proof}
Along Definition~\ref{def:thierarchical} it is straightforward to
construct an algorithm which decides in time $\poly({q})$ whether
a given CQ $q$ is t-hierarchical.

Let $q\deff q_\phi(z_1,\ldots,z_k)$ be a given t-hierarchical CQ.
Let $A_0$ be the set of all atoms $\psi$ of $q$ with
$\vars(\psi)\subseteq\free(q)$, and let $\phi_0$ be the
quantifier-free conjunctive formula
\[
  \phi_0 \ \deff \ \ \Und_{\psi\in A_0} \psi\,.
\]
For each $Z\subseteq \free(q)$ let $A_Z$ be the set of all atoms
$\psi$ of $q$ such that $\vars(\psi)\varsupsetneq
\vars(\psi)\cap\free(q) = Z$.
Let $Z_1,\ldots,Z_n$ (for $n\geq 0$) be a list of all those
$Z\subseteq \free(q)$ with $A_Z\neq \emptyset$.
For each $j\in[n]$ let $A_j\deff A_{Z_j}$ and let 
$Y_j\deff \big( \bigcup_{\psi\in A_j} \vars(\psi) \big)\setminus Z_j$.

\smallskip

\begin{claim}\label{claim1:lemma:t-hierCQ-Decomp}
$Y_j\cap Y_{j'}=\emptyset$ for all $j,j'\in[n]$ with $j\neq j'$.
\end{claim}
\begin{proof}
We know that $Z_j\neq Z_{j'}$. W.l.o.g.\ there is a $z\in Z_j$ with
$z\not\in Z_{j'}$.

For contradiction, assume that $Y_j\cap Y_{j'}$ contains some variable
$y$.
Then, $y\in\vars(\psi)$ for some $\psi\in A_{j}$ and
$y\in\vars(\psi')$ for some $\psi'\in A_{j'}$.
By definition of $A_j$ we know that $\vars(\psi)\cap\free(q)=Z_j$, and
hence $z\in\vars(\psi)$.
By definition of $A_{j'}$ we know that $\vars(\psi')\cap\free(q)=Z_{j'}$, and
hence $z\not\in\vars(\psi')$.
Hence, $\psi\in\atoms(z)$ and $\psi'\not\in\atoms(z)$.
Since $\psi\in\atoms(y)$ and $\psi'\in\atoms(y)$, we obtain that
$\atoms(z)\cap\atoms(y)\neq\emptyset$ and
$\atoms(y)\not\subseteq\atoms(z)$.
But by assumption, $q$ is t-hierarchical, and this contradicts
condition~\eqref{eq:t-hier-cond-ii} of 
Definition~\ref{def:thierarchical}.
\end{proof}
For each $j\in[n]$ consider the
conjunctive formula
\[
 \phi_j \ \deff \ \ 
 \exists y_1^{(j)}\cdots\exists y_{\ell_j}^{(j)}\ \Und_{\psi\in A_{j}} \psi\,,
\]
where
$\ell_j\deff |Y_j|$ and
$(y_1^{(j)},\ldots,y_{\ell_j}^{(j)})$ is a list of all variables in $Y_j$.
Using Claim~\ref{claim1:lemma:t-hierCQ-Decomp}, it is straightforward
to see that
\[
  q' \ \deff \ \ 
  \{\
    (z_1,\ldots,z_k) \ : \ 
    \phi_0 \uund \Und_{j\in[n]}\phi_j
  \ \}
\]
is a generalised CQ that is equivalent to $q$.
Furthermore, $q'$ can be constructed in time $\poly({q})$.
To complete the proof of Lemma~\ref{lemma:t-hierCQ-Decomp} we
consider for 
each $j\in[n]$ the CQ
\[
  q_j \ \deff \ \
  \{\ 
   \ov{z}^{(j)} \ : \ \phi_j
  \ \}\,,
\] 
where $\ov{z}^{(j)}$ is a tuple of length $|Z_j|$
consisting of all the variables in $Z_j$.

\smallskip

\begin{claim}\label{claim2:lemma:t-hierCQ-Decomp}
$q_j$ is q-hierarchical, for each $j\in[n]$.
\end{claim}

\begin{proof}
First of all, note that $q_j$ satisfies condition~\eqref{eq:q-hier-cond-ii} of
Definition~\ref{def:qhierarchical}, since
$\free(q_j)=Z_j$, $\atoms_{q_j}(z)=A_j$ for every $z\in Z_j$, and
$\atoms_{q_j}(y)\subseteq A_j$ for every $y\in Y_j=\vars(q_j)\setminus\free(q_j)$.

For contradiction, assume that $q_j$ is not q-hierarchical.
Then, $q_j$ violates condition~\eqref{eq:q-hier-cond-i} of
Definition~\ref{def:qhierarchical}. I.e.,
there are variables $x,x'\in Z_j\cup Y_j$ and atoms $\psi_1,
\psi_2, \psi_3 \in A_j$ such that
$\vars(\psi_1)\cap \set{x,x'}=\set{x}$,
$\vars(\psi_2)\cap\set{x,x'}=\set{x'}$, and
$\vars(\psi_3)\cap\set{x,x'}=\set{x,x'}$.
Since $\vars(\psi)\cap\free(q)=Z_j$ for all $\psi\in A_j$, we
know that $x,x'\not\in \free(q)$. Therefore,
$x,x'\in\vars(q)\setminus\free(q)$, and hence $\psi_1$, $\psi_2$,
$\psi_3$ are atoms of $q$ which witness that condition~\eqref{eq:t-hier-cond-i} of 
Definition~\ref{def:thierarchical} is violated. This contradicts the
assumption that $q$ is t-hierarchical.
\end{proof}
\noindent
This completes the proof of Lemma~\ref{lemma:t-hierCQ-Decomp}.
\end{proof}

\noindent
The proof of
Theorem~\ref{thm:testing}\,\eqref{item:thm:testing:upperbound} now
follows easily: When given a t-hierarchical CQ
$q_\phi(z_1,\ldots,z_k)$, use the algorithm provided by
Lemma~\ref{lemma:t-hierCQ-Decomp} to compute an equivalent generalised CQ $q'$ of
the form $\{ (z_1,\ldots,z_k) \ : \  \phi_1\und\cdots\und \phi_m\}$
and let $q_j\deff \{\ov{z}^{(j)} \ : \ \phi_j \}$ for
each $j\in[m]$.
W.l.o.g.\ assume that there is an $m'\in\set{0,\ldots,m}$ such that 
$q_j$ is q-hierarchical for each $j\leq m'$ and $q_j$ is
quantifier-free for each $j>m'$.
We use in parallel, for each $j\leq m'$, the data structures provided
by Theorem~\ref{thm:CQs}\,\eqref{item:thm:CQs:upper} for the
q-hierarchical CQ $q_j$. 
In addition to this, we use an $r$-dimensional array
$\arrayfont{A}_R$ for each relation symbol $R\in\schema$ of arity
$r\deff\ar(R)$, and we ensure
that for all $\ov{b}\in \Dom^{r}$ we have $\arrayfont{A}_R[\ov{b}]=1$
if $\ov{b}\in R^\DB$, and $\arrayfont{A}_R[\ov{b}]=0$ otherwise.
When receiving an update command $\Update\,R(\ov{b})$, 
we let $\arrayfont{A}_R[\ov{b}]\deff 1$ if $\Update=\Insert$, and
$\arrayfont{A}_R[\ov{b}]\deff 0$ if $\Update=\Delete$, and in addition to
this, we call
the $\UPDATE$ routines of the data structure for $q^{(j)}$ for each $j\leq m'$.
Upon input of a tuple $\ov{a}\in\Dom^k$, the $\TEST$ routine 
proceeds as follows.
For each $j\leq m'$, it calls the $\TEST$ routine of the data
structure for $q^{(j)}$ upon input $\ov{a}^{(j)}$. And additionally,
it uses the arrays $\arrayfont{A}_R$ for all $R\in\schema$ to
check if for each $j>m'$ the quantifier-free query $q^{(j)}$ is
satisfied by the tuple $\ov{a}^{(j)}$. All this is done within time
$\poly({q})$, and we know that $\ov{a}\in q(\DB)$ if, and only
if, all these tests succeed.
This completes the proof of part \eqref{item:thm:testing:upperbound} of Theorem~\ref{thm:testing}.

\bigskip

Let us now turn to the proof of part \eqref{item:thm:testing:lowerbound} of Theorem~\ref{thm:testing}.
We are given a query $q \deff q_\phi(z_1,\ldots,z_k)$ and without loss
of generality we assume that $q$ is a homomorphic
core and $q$ is not t-hierarchical.
Thus, $q$ violates condition \eqref{eq:t-hier-cond-i} or
\eqref{eq:t-hier-cond-ii} of Definition~\ref{def:thierarchical}.
In case that it violates condition \eqref{eq:t-hier-cond-i}, the proof
is virtually identical to the proof of Theorem~3.4 in
\cite{BKS_enumeration_PODS17}; for the reader's convenience, the proof
details are given in Appendix~\ref{appendix:TestingLowerBound}.

Let us consider the case where $q$ violates
condition~\eqref{eq:t-hier-cond-ii} of Definition~\ref{def:thierarchical}.
In this case, there are two variables $x\in\free(q)$ and
$y\in\Vars(q)\setminus\free(q)$
and two atoms
$\sgpsixy$ and $\sgpsiy$ of $q$ 
with
$\Vars(\sgpsixy)\cap\set{\varx,\vary}=\set{\varx,\vary}$ and
$\Vars(\sgpsiy)\cap\set{\varx,\vary}=\set{\vary}$.  
The easiest example of a query for which this is true is 
\ $
\qET \deff  \{ \,
  (x) \, : \, 
  \exists y\, (\, Exy \uund Ty\,)
\, \}\,.
$ \
Here, we illustrate the proof idea for the particular query
$\qET$; a proof for the general case is given in Appendix~\ref{appendix:TestingLowerBound}.

Assume that there is a dynamic algorithm that solves
the testing problem for $\qET$ with update time
$\updatetime=\bigOh(n^{1-\epsilon})$ and testing time
$\testingtime=\bigOh(n^{1-\epsilon})$ on databases whose active domain
is of size $\bigOh(n)$.
We show how this algorithm can be used to
solve the OuMv-problem.

For the OuMv-problem, we receive as input
an $n\times n$ matrix $M$. We start the
preprocessing phase of our testing algorithm for $\qET$ with the
empty database $\DB=(E^\DB,T^\DB)$ where
$E^\DB=T^\DB=\emptyset$.
As this database has constant size, the preprocessing is finished in
constant time.
We then apply $\bigOh(n^2)$ update steps to ensure that
$E^\DB=\setc{(i,j)}{M_{i,j}=1}$. 
All this takes time at most $\bigOh(n^2)\cdot
\updatetime = \bigOh(n^{3-\epsilon})$.
Throughout the remainder of the construction, we will never change
$E^{\DB}$, and we will always ensure that 
$T^{\DB}\subseteq [n]$.

When we receive two vectors $\vec{u}^{\,t}$ and $\vec{v}^{\,t}$ in the
dynamic phase of the OuMv-problem, we proceed as follows. First, 
we perform the update commands $\Delete\,T(j)$ for each $j\in [n]$
with $\vec{v}^{\,t}=0$, and
the update commands $\Insert\,T(j)$ for each $j\in [n]$
with $\vec{v}_j^{\,t}=1$.
This is done within time $n\cdot\updatetime
=\bigOh(n^{2-\epsilon})$.
By construction of $\DB$ we know that for every $i\in[n]$ we have
\[
 i\ \in\ \qET(\DB)
 \quad \iff \quad
 \text{%
  there is a $j\in[n]$ such that $M_{i,j}=1$ and
  $\vec{v}_j^{\,t}=1$\,.
 }%
\]
Thus,
\ $
  (\vec{u}^{\,t})\trans M \vec{v}^{\,t} = 1
  \iff 
  \text{%
   there is an $i\in[n]$ with $\vec{u}_i^{\,t}=1$ and $i\in\qET(\DB)$.
  }%
$\,
Therefore, after having called the $\TEST$ routine for $\qET$ for each
$i\in[n]$ with $\vec{u}_i^{\,t}=1$, we can output the correct result
of $(\vec{u}^{\,t})\trans M \vec{v}^{\,t}$.
This takes time at most $n\cdot \testingtime=\bigOh(n^{2-\epsilon})$.
I.e., for each $t\in[n]$ after receiving the vectors $\vec{u}^{\,t}$
and $\vec{v}^{\,t}$, we can output $(\vec{u}^{\,t})\trans
M\vec{v}^{\,t}$ within time $\bigOh(n^{2-\epsilon})$. 
Consequently, the overall running time for solving the OuMv-problem is bounded by $\bigOh(n^{3-\epsilon})$.

By using the technical machinery of \cite{BKS_enumeration_PODS17},
this approach can be generalised from $\qET$ to all queries $q$ that 
violate condition~\eqref{eq:t-hier-cond-ii} of
Definition~\ref{def:thierarchical}; see
Appendix~\ref{appendix:TestingLowerBound} for details.
This completes the proof of Theorem~\ref{thm:testing}.
\end{proof}

\section{Unions of conjunctive queries}\label{section:UCQs}

In this section we consider dynamic query evaluation for UCQs.
To transfer our notions of \emph{hierarchical} queries from CQs to
UCQs, we say that a UCQ  $q(\ov{u})$ of the form $q_{1}(\ov{u}_1)
\,\cup\,\cdots \,\cup \,q_{d}(\ov{u}_d)$
is \qhier (t-hierarchical) if every CQ $q_{i}(\ov{u}_i)$ in
the union is \qhier (t-hierarchical).
Note that for \emph{Boolean} queries (CQs as well as UCQs) the notions
of being q-hierarchical and being t-hierarchical coincide,
and for a $k$-ary UCQ $q$ it can be checked in time $\poly(q)$ if
$q$ is q-hierarchical or t-hierarchical.

The following theorem generalises the statement of
Theorem~\ref{thm:testing} from CQs to UCQs.
Its proof follows 
easily from the Theorems~\ref{thm:testing} and \ref{thm:CQs}.

\begin{theorem}\label{thm:UCQtesting}
\begin{enumerate}[(a)]
\item\label{item:thm:UCQtesting:upperbound}
There is a dynamic algorithm that receives a t-hierarchical $k$-ary
UCQ $q$ and a $\schema$-db $\DBstart$, and computes
within $\preprocessingtime =\poly({q})\cdot\bigOh(\size{\DBstart})$ preprocessing time a data structure that can be
updated in time $\updatetime = \poly({q})$ and allows to test for an
input tuple $\ov{a}\in\Dom^k$ if $\ov{a}\in q(\DB)$ within time 
$\testingtime = \poly({q})$. Furthermore, the algorithm allows to
answer a q-hierarchical Boolean UCQ within time $\answertime = \bigOh(1)$.
\item\label{item:thm:UCQtesting:lowerbound}
Let $\epsilon>0$ and let $q$ be a $k$-ary UCQ whose
homomorphic core is not t-hierarchical 
(note that this is the case if, and only
  if, $q$ is not equivalent to a t-hierarchical UCQ).
There is no dynamic algorithm with arbitrary preprocessing time and $\updatetime =\bigOh(n^{1-\epsilon})$
update time that can test for any input tuple $\ov{a}\in\Dom^k$ if
$\ov{a}\in q(\DB)$ within testing time
$\testingtime=\bigOh(n^{1-\epsilon})$, unless the OMv-conjecture fails.
Furthermore, if $k=0$ (i.e., $q$ is a Boolean UCQ), then there 
is no dynamic algorithm with arbitrary preprocessing time and
$\updatetime=\bigoh(\actdomsize^{1-\smalleps})$
update time that
answers 
$\eval{q}{\DB}$ in time
$\answertime=\bigoh(\actdomsize^{2-\smalleps})$, unless the \OMvcon
fails.
\end{enumerate}
\end{theorem}

\begin{proof}
  Part~\eqref{item:thm:UCQtesting:upperbound} 
  follows immediately from 
  Theorem~\ref{thm:testing}\,\eqref{item:thm:testing:upperbound} (and
  Theorem~\ref{thm:CQs}\,\eqref{item:thm:CQs:count} for the statement on
  \emph{Boolean} UCQs),
  as we
  can maintain all CQs in the union in parallel and then decide
  whether at least one of them is satisfied by the current database
  and the given tuple.

  For the proof of \eqref{item:thm:UCQtesting:lowerbound} let $q'$ be
  the homomorphic core of $q$, and let $q'$ be of the form
  $q_1\cup\cdots\cup q_m$ for $k$-ary CQs $q_1,\ldots,q_m$.
  We first consider the case that $q$ is a \emph{Boolean} UCQ.
  Then, by assumption, $q'$ is not q-hierarchical, and hence there
  exists an $i\in[m]$ such that the Boolean CQ $q_{i}$ a
  is not \qhier. Suppose for
  contradiction that there exists a dynamic algorithm that evaluates
  $q$ with $\updatetime=\bigoh(\actdomsize^{1-\smalleps})$
  update time and
  $\answertime=\bigoh(\actdomsize^{2-\smalleps})$ answer time.
  Since $q'$ is equivalent to $q$, it
  follows that the algorithm also evaluates $q'$ with the same time
  bounds. Now consider the class of databases that map homomorphically
  into $q_i$.
  For every database $\DB$ in this class it holds that $q'(\DB)=\Yes$ if, and
  only if, $q_i(\DB)=\Yes$. This is because every other CQ $q_j$ in
  $q'$ is not satisfied by $\DB$, since otherwise there would be a
  homomorphism from $q_j$ to $\DB$ and therefore, since there is a homomorphism
  from $\DB$ to $q_i$, also from $q_j$ to $q_i$, contradicting that
  $q'$ is a core.
  Hence, the dynamic algorithm evaluates the non-q-hierarchical CQ
  $q_i$ (which is a homomorphic core), contradicting
  Theorem~\ref{thm:CQlower}\,\eqref{item:thm:CQlower:answering}.

  The statement of \eqref{item:thm:UCQtesting:lowerbound} concerning
  non-Boolean UCQs and the testing problem follows along the same
  lines when using 
  Theorem~\ref{thm:testing}\,\eqref{item:thm:testing:lowerbound}
  instead of
  Theorem~\ref{thm:CQlower}\,\eqref{item:thm:CQlower:answering}.
  \end{proof}
It turns out that 
similarly as \qhier CQs, also
\qhier UCQs allow for efficient enumeration under
updates. This, and the according lower bound, is stated in the
following Theorem~\ref{thm:enumUCQ}, which will be proven 
at the end of this section.
In contrast to Theorem~\ref{thm:UCQtesting}, the result does not follow
immediately from the tractability of the enumeration problem for \qhier CQs, because one has
to ensure that tuples from result sets of two different CQs are not
reported twice while enumerating their union.

\begin{theorem}\label{thm:enumUCQ}
\begin{enumerate}[(a)]
\item\label{item:thm:enumUCQ:upper}
  There is a dynamic algorithm that receives a q-hierarchical $k$-ary
  UCQ $q$
  and a $\schema$-db $\DBstart$, and computes within
  $\preprocessingtime =\poly(q)\cdot\bigOh(\size{\DBstart})$
  preprocessing time a data structure that can be updated in time
  $\updatetime = \poly(q)$ and allows to 
  enumerate $q(\DB)$ with delay $\delaytime = \poly(q)$.

\item\label{item:thm:enumUCQ:lower}
  Let $\epsilon>0$ and let $q$ be a $k$-ary UCQ whose homomorphic core
  is not \qhier and
  is a 
  union of self-join free CQs.
  There is no dynamic
  algorithm with arbitrary preprocessing time and
  $\updatetime=\bigoh(\actdomsize^{1-\smalleps})$ update time that enumerates 
  $\eval{q}{\DB}$ with
  delay $\delaytime=\bigoh(\actdomsize^{1-\smalleps})$, unless the \OMvcon
  fails. 
\end{enumerate}
\end{theorem}

Note that according to Theorem~\ref{thm:CQs}, for CQs the enumeration
problem as well as the counting problem can be solved by efficient
dynamic algorithms if, and (modulo algorithmic conjectures) only if, the query is \qhier.
In contrast to this, it turns out that for UCQs computing the number
of output tuples can be much harder than enumerating the query
result. 
To characterise the UCQs that allow for efficient dynamic counting algorithms, we
use the following notation.
For two $k$-ary CQs $q_{\phi}(u_1,\ldots,u_k)$ and
$q_{\psi}(v_1,\ldots,v_k)$ we define the intersection $q \deff
q_{\phi}\cap q_{\psi}$ to be the following $k$-ary query. If
there is an $i\in[k]$ such that $u_i$ and $v_i$ are distinct elements
from $\Dom$, then $q := \emptyset$ (and this query is q-hierarchical
by definition). Otherwise, we let $w_1,\ldots,w_k$ be elements from
$\Var\cup\Dom$ which satisfy the following for all $i,j\in[k]$ and all
$a\in\Dom$:
\begin{equation*}
\text{
 $\big($ $w_i=a$ $\iff$ $u_i=a$ \,or \,$v_i=a$ $\big)$
 \quad and \quad 
 $\big($ $w_i=w_j$ $\iff$  $u_i=u_j$ \,or \,$v_i=v_j$ $\big)$.
}%
\end{equation*}
We obtain $\phi'$ from $\phi$ (and $\psi'$ from $\psi$) by replacing every
$u_i\in\{u_1,\ldots, u_k\}\cap \free(\phi)$ (and $v_i\in\{v_1,\ldots, v_k\}\cap \free(\psi)$) by $w_i$.
Finally, we let $q =  \{ \ (w_1,\ldots,w_k) \ : \
\phi'\,\wedge\,\psi' \ \}$, where we can assume \mbox{w.l.o.g.} that
$\phi'\,\wedge\,\psi'$ is a conjunctive formula of the form \eqref{eq:CF}.
Note that for every database $\DB$ it holds that $q(\DB) =
q_{\phi}(\DB)\cap q_{\psi}(\DB)$.

\begin{definition}\label{def:exhaustively-qhier}
A UCQ $q$ of the form
$\bigcup_{i\in[d]}q_{i}(\ov{u}_i)$
 is \emph{exhaustively \qhier} if for every
$I\subseteq [d]$ the intersection
$q_{I}=\bigcap_{i\in I} q_{i}$ is equivalent to a \qhier CQ.
\end{definition}

It is not difficult to see that 
a \emph{Boolean} UCQ is exhaustively \qhier if and only if its
homomorphic core is
\qhier.
In the non-Boolean case, being exhaustively \qhier is a stronger
requirement than being \qhier as the following example shows: the UCQ
\ $\{\,
  (x,y) \, : \,
    Sx \uund Exy
    \, \}
    \  \cup\ 
      \{\,
  (x,y) \, : \,
    Exy \uund Ty
  \, \} 
$ \
is \qhier, but not exhaustively \qhier.
In contrast to the \qhier property, the straightforward way of deciding whether a UCQ $q$ is exhaustively
\qhier requires $2^{\poly(q)}$ and it is open whether this can be improved.
The next theorem shows
that the exhaustively \qhier queries are precisely those UCQs that allow for efficient dynamic
counting algorithms. 

\begin{theorem}\label{thm:countUCQ}
\begin{enumerate}[(a)]
\item\label{item:thm:countUCQ:upper}
  There is a dynamic algorithm that receives an exhaustively
  q-hierarchical UCQ $q$
  and a $\schema$-db $\DBstart$, and computes within 
  $\preprocessingtime =2^{\poly(q)}\cdot\bigOh(\size{\DBstart})$
  preprocessing time a data structure that can be updated in time
  $\updatetime = 2^{\poly(q)}$ 
  and computes $\setsize{q(\DB)}$ in time $\countingtime=\bigOh(1)$.

\item\label{item:thm:countUCQ:lower}
  Let $\epsilon>0$ and let $q$ be a UCQ
  whose homomorphic core is not exhaustively \qhier.
  There is no dynamic algorithm with arbitrary preprocessing time and
  $\updatetime=\bigoh(\actdomsize^{1-\smalleps})$ update time that computes
  $\setsize{\eval{q}{\DB}}$ 
  in time $\countingtime=\bigoh(\actdomsize^{1-\smalleps})$,
  unless the
  \OMvcon or the \OVcon fails.
\end{enumerate}
\end{theorem}

\newcommand{\coreof}[1]{\widetilde{#1}}
\newcommand{\ucq}{\query}
\newcommand{\cq}{\query}
\newcommand{\cqpsi}{\psi}
\newcommand{\IndSet}{{I}}
\newcommand{\IndSetJ}{{J}}
\newcommand{\IndSets}{{\mathcal I}}
\newcommand{\exqhier}{exhaustively q-hie\-rar\-chi\-cal\xspace}

\begin{proof}
To prove part~\eqref{item:thm:countUCQ:upper} we use the principle of
inclusion-exclusion along with the upper bound of
Theorem~\ref{thm:CQs}\,\eqref{item:thm:CQs:count}.
  Let $q=\bigcup_{i\in[d]}q_{i}(\ov{u}_i)$ be an
  exhaustively \qhier UCQ.
  Our dynamic algorithm for solving the counting problem for $q$
  proceeds as follows.
  In the preprocessing phase we first compute for every non-empty
  $I\subseteq[d]$ the homomorphic core $\coreof{q_I}$ of the CQ
  $q_{I}\deff\bigcap_{i\in I}q_{i}$. This can be done in time
  $2^{\poly(q_{I})}$.
  Since $q$ is exhaustively \qhier, every $\coreof{q_{I}}$ is \qhier
  and we can apply Theorem~\ref{thm:CQs}\,\eqref{item:thm:CQs:count} to
  determine the number of result tuples $\setsize{\coreof{q_{I}}(\DB)}=\setsize{q_{I}(\DB)}$ for
  every $I\subseteq[d]$
  with $\sum_{I\subseteq[d]}\poly(q_{I}) = 2^{\poly(q)}$ update time.
  By the principle of inclusion-exclusion we have
  \begin{equation*}
    \setsize{q(\DB)} 
    \ \ = \ \
    |\bigcup_{i\in[d]} q_i(\DB)|
    \ \ = \ \
    \sum_{\emptyset\neq I\subseteq
      [d]}(-1)^{\setsize{I}+1}\,{\cdot}\,\setsize{\bigcap_{i\in I}q_{i}(\DB)}
    \ \ = \ \
    \sum_{\emptyset\neq I\subseteq [d]}(-1)^{\setsize{I}+1}\,{\cdot}\,\setsize{q_{I}(\DB)}.
  \end{equation*}
  Therefore, we can compute the number of result tuples in
  $q(\DB)$ by maintaining all $2^d-1$ numbers
  $\setsize{\coreof{q_{I}}(\DB)}$ in parallel (for all non-empty
  $I\subseteq[d]$).

For the proof of part~\eqref{item:thm:countUCQ:lower}
let $\coreof{q}=\bigcup_{i\in[d]}q_{i}(\ov{u}_i)$ be the
 homomorphic core of $q$. Consider the CQs $q_{I}\deff\bigcap_{i\in
    I}q_{i}$ and their homomorphic cores $\coreof{q_I}$ for all non-empty $I\subseteq [d]$. First we take care of equivalent queries and write $I \cong
  J$ if  $q_{I}\equiv q_{J}$.
  Let $\IndSets$ be a set of index sets $I$ that contains one
  representative from each equivalence class
  $I/_{\cong}$. By the principle of inclusion-exclusion we have
  \begin{equation*}
   \setsize{q(\DB)} 
   \ \ = \ \ 
   \setsize{\coreof{q}(\DB)} 
   \ \ = \ \  
   \sum_{\emptyset\neq\IndSet\subseteq
     [d]}(-1)^{\setsize{\IndSet}+1}\,{\cdot}\,\setsize{q_{\IndSet}(\DB)}
   \ \ = \ \ 
   \sum_{\IndSet\in\IndSets}a_{\IndSet}\,{\cdot}\,\setsize{q_{\IndSet}(\DB)}\,,
  \end{equation*}
  where
  $a_{\IndSet}\defi \sum_{\IndSetJ\colon
    \IndSetJ\cong\IndSet}(-1)^{\setsize{\IndSetJ}+1}$. Because $q$ is
  not exhaustively \qhier, we can
  choose a set $\IndSet\in\IndSets$ such that $\coreof{q}_{\IndSet}$ is a
  non-q-hierarchical query, which is \emph{minimal} in the sense that for
  every $\IndSetJ\in \IndSets\setminus \{\IndSet\}$ there is no
  homomorphism from ${q}_\IndSetJ$ to ${q}_\IndSet$. Note that such a
  minimal set $I$ exists since otherwise we could find two distinct
  $J, J' \in \IndSets$ such that
  $q_\IndSetJ \equiv q_{\IndSetJ'}$.

  Now suppose that $\DB$ is a database from the class of databases
  that map homomorphically into $q_I$ and let $h\colon \DB\to q_\IndSet$
  be a homomorphism. For every
  $\IndSetJ\in\IndSets\setminus\set{\IndSet}$ it holds that there is no
  homomorphism $h'\colon q_\IndSetJ\to\DB$, since otherwise $h\circ
  h'$ would be a homomorphism from $q_\IndSetJ$ to $q_\IndSet$.
  Hence, $q_\IndSetJ(\DB)=\emptyset$ for all
  $\IndSetJ\in\IndSets\setminus\set{\IndSet}$ and thus
  $\setsize{q(\DB)}=a_\IndSet\cdot\setsize{q_\IndSet(\DB)}$. It
  follows that we can compute
  $\setsize{q_\IndSet(\DB)}=\setsize{\coreof{q}_\IndSet(\DB)}$  
  by maintaining the value for
  $\setsize{q(\DB)}$ and dividing it by $a_\IndSet$.
  Since $\coreof{q}_\IndSet$ is a non-\qhier homomorphic core, the lower bound for
  maintaining $\setsize{q(\DB)}$ follows from
  Theorem~\ref{thm:CQs}\,\eqref{item:thm:CQlower:counting}.
This completes the proof of Theorem~\ref{thm:countUCQ}.
\end{proof}

\newcommand{\setT}{T}
\newcommand{\elemt}{t}
\newcommand{\nextelement}{\mathsf{next}}
\newcommand{\startelement}{\mathsf{start}}
\newcommand{\skipforth}{\mathsf{skip}}
\newcommand{\skipback}{\mathsf{skipback}}

The remainder of this section is devoted to the proof of  Theorem~\ref{thm:enumUCQ}.
To prove
Theorem~\ref{thm:enumUCQ}\,\eqref{item:thm:enumUCQ:upper}, we first develop a general
method for enumerating the union of sets.
We say that a data structure for a set $T$ \emph{allows to skip}, if it is
possible to test whether $t\in T$ in constant time and for some 
ordering $t_1,\ldots,t_n$ of the elements in $T$ there is
\begin{itemize}
\item a function $\startelement$, which returns $t_1$ in constant time
  and
\item a function $\nextelement(t_i)$, which returns $t_{i+1}$ (if
  $i<n$) or \EOE (if $i=n$) in constant time.
\end{itemize}

Note that a data structure that allows to skip enables constant
delay enumeration of $t_i$, $t_{i+1}$, \ldots, $t_n$ starting from an
arbitrary element $t_i\in T$ (but we do
not have control over the underlying order).
An example of such a data structure is an explicit
representation of the elements of $T$ in a linked list with constant access.
Another example is the data structure of the enumeration
algorithm for the result $\setT\deff q(\DB)$ of a \qhier
CQ $q$, provided by Theorem~\ref{thm:CQs}\,\eqref{item:thm:CQs:enumerate}\&\eqref{item:thm:CQs:next}.
The next lemma states that we can use these data structures for sets
$\setT_j$ to enumerate the union $\bigcup_j \setT_j$ with constant delay and
without repetition.

\begin{lemma}\label{lem:union}
 Let $\ell\geq 1$ and
 let $\setT_1,\ldots,\setT_\ell$ be sets such that for each
 $j\in[\ell]$ there is a data structure for $\setT_j$ that allows to skip.
 Then there is an algorithm that enumerates, without repetition, all
 elements in
 $\setT_1\cup\cdots\cup\setT_\ell$ with $\bigOh(\ell)$ delay.
\end{lemma}
\begin{proof}
  For each $i\in[\ell]$ let $\startelement^i$ and $\nextelement^i$ be the start element and
  the iterator for the set $T_i$.  The main idea for enumerating the
  union $\setT_1\cup\cdots\cup\setT_\ell$ is to first enumerate all
  elements in $T_1$, and then $T_2\setminus T_1$, \
  $T_3\setminus (T_1\cup T_2)$, \ \ldots, \
  $T_\ell\setminus (T_1\cup\cdots \cup T_{\ell-1})$.  In order to do
  this we have to exclude all elements that have already been
  reported from all subsequent sets. As we want to ensure constant
  delay enumeration, we cannot just ignore the elements in $T_i \cap
  (T_1\cup\cdots \cup T_{i-1})$ while enumerating $T_i$. As a remedy, we
  use an additional
  pointer to jump from an element that has already been reported
  to the least element that needs to be reported next.
  To do this we use 
  arrays $\skipforth^i$ (for all $i\in[\ell]$) to jump over excluded elements: if
  $\elemt_r,\ldots,\elemt_s$ is a maximal interval of elements in $\setT_i$ that
  have already been reported, then
  $\skipforth^i[\elemt_r]=\elemt_{s+1}$ (if
  $\elemt_s$ is the last element in $\setT_i$, then $\elemt_{s+1} :=
  \EOE$). For technical reasons we also need the array
  $\skipback^i$ which represents the inverse pointer, i.e., $\skipback^i[\elemt_{s+1}]=\elemt_{r}$.

  The enumeration algorithm is stated in
  Algorithm~\ref{alg:enum-set-disjunction}. It uses the procedure
  \textsc{exclude}$^j$ described in Algorithm~\ref{alg:delete} to update the
  arrays whenever an element $t$ has been reported. 
It is straightforward to verify that these algorithms provide the
desired functionality within the claimed time bounds.
\end{proof}

\begin{algorithm}[h!tb]
 \caption{The enumeration algorithm for $\setT_1\cup\cdots\cup\setT_\ell$}\label{alg:enum-set-disjunction} 
 \begin{algorithmic}
  \State \textbf{Input:} Data structures for sets $\setT_j$ with first
  element
  $\startelement^j$ and iterator $\nextelement^j$. 
  \State Pointer $\skipforth^j[t] = \skipback^j[t] = \nil$
  for all $j\in[\ell]$ and $t\in T_j$. 
  \State
  \For{$i = 1,\, \ldots,\, \ell$}
  \State $t = \startelement^i$
  \While{$t \neq \EOE$}
  \If{$\skipforth^i[t] == \nil$}
  \State Output element $t$
  \For{$j = i+1\ \to\ \ell$}
  \State \textsc{exclude}$^j(t)$ 
  \EndFor
  $t = \nextelement^i(t)$
  \Else
  \State $t = \skipforth^i[t]$
  \EndIf
  \EndWhile
  \EndFor
  \State Output the end-of-enumeration message $\EOE$.
 \end{algorithmic}
\end{algorithm}

\begin{algorithm}[h!tb]
 \caption{Procedure \textsc{exclude}$^j$ for excluding $\elemt$ from $\setT_j$}\label{alg:delete}
 \begin{algorithmic}
   \If{$\elemt \in \setT_j$}
     \If{$\skipback^j[\elemt] \neq \nil$}
     \State $\elemt^- = \skipback^j[\elemt]$
     \State $\skipback^j[\elemt] = \nil$
     \Else
     \State $\elemt^- = \elemt$ 
     \EndIf
      \If{$\skipforth^j[\nextelement^j(\elemt)] \neq \nil$}
      \State $\elemt^+ = \skipforth^j[\nextelement^j(\elemt)]$
      \State $\skipforth^j[\nextelement^j(\elemt)] = \nil$
     \Else
         \State $\elemt^+=\nextelement^j(\elemt)$
     \EndIf
     \State $\skipforth^j[\elemt^-] = \elemt^+$;\quad $\skipback^j[\elemt^+] = \elemt^-$
   \EndIf
 \end{algorithmic}
\end{algorithm}

Lemma~\ref{lem:union} enables us to prove the upper bound of
Theorem~\ref{thm:enumUCQ}, and the lower bound is proved by using
Theorem~\ref{thm:CQs}\,\eqref{item:thm:CQs:enumeration}.

\begin{proof}[Proof of Theorem~\ref{thm:enumUCQ}]
  The upper bound follows immediately from combining Lemma~\ref{lem:union} with
  Theorem~\ref{thm:CQs}\,\eqref{item:thm:CQs:next}.
  For the lower bound let $q_i$ be a self-join free non-\qhier CQ in
  the homomorphic core $q'$ of the UCQ $q$. For every database $\DB$ that maps
  homomorphically into $q_i$ it holds that $q_j(\DB)=\emptyset$ for
  every other CQ $q_j$ in $q'$ (with $j\neq i$), since otherwise there
  would be a homomorphism from $q_j$ to $\DB$ and hence to $q_i$,
  contradicting that $q'$ is a homomorphic core.
  It follows that every dynamic algorithm that enumerates the result
  of $q$ on a
  database $\DB$ which maps homomorphically into $q_i$ also enumerates
  $q_i(\DB)=q(\DB)$, contradicting Theorem~\ref{thm:CQs}\,\eqref{item:thm:CQs:enumeration}.
\end{proof}

\section{CQs and UCQs with integrity constraints}\label{section:QueriesWithIntegrityConstraints}

In the presence of integrity constraints, the characterisation of
tractable queries changes and depends on the query as well as on the
set of constraints.
When considering a scenario where databases are required to satisfy a
set $\CONSTR$ of constraints,  
we allow to execute a given $\UPDATE$ command only if the resulting 
database still satisfies all constraints in $\CONSTR$.
When speaking of \emph{$(\schema,\CONSTR)$-dbs} we mean $\schema$-dbs
$\DB$ that satisfy all constraints in $\CONSTR$.
Two queries $q$ and $q'$ are 
\emph{$\CONSTR$-equivalent} (for short: $q\equiv_\CONSTR q'$) if
$q(\DB)=q'(\DB)$ for every 
$(\schema,\CONSTR)$-db $\DB$.

We first consider \emph{small domain constraints}, i.e.,
constraints $\DEP$ of the form
\;$
  R[i]\subseteq C
$\;
where $R\in\schema$, $i\in\set{1,\ldots,\ar(R)}$, and $C\subseteq\Dom$
is a finite set.
A $\schema$-db $\DB$ \emph{satisfies} $\DEP$ if
$\proj_i(R^\DB)\subseteq C$.

For these constraints we are able to give a clear picture of
the tractability landscape by reducing CQs and UCQs with small domain
constraints to UCQs without integrity constraints and applying the
characterisations for UCQs achieved in Section~\ref{section:UCQs}.
We start with an example that illustrates how
a query can be simplified in the presence of
small domain constraints.

\begin{example}\label{example:SD-constraints}
Consider the Boolean query  
\ $
 \qSET \deff 
 \{ \, 
   \emptytuple \, : \, 
   \exists x \exists y \, ( \, Sx \uund Exy\uund Ty\,)
 \, \}\,,
$ \
which is not \qhier. By Theorem~\ref{thm:CQs} 
it cannot be answered by a dynamic algorithm with sublinear update
time and sublinear answer time, unless the OMv-conjecture fails. 
But in the presence of the small domain constraint
\,$
 \SD \deff S[1]\subseteq C
$
for a set $C\subseteq\Dom$ of the form $C=\set{a_1,\ldots,a_c}$, the
query $\qSET$ is $\set{\SD}$-equivalent to the q-hierarchical UCQ
\[
  q' \ \deff \ \ 
  \bigcup_{a_i\in C} \ 
  \{ \ 
     \emptytuple \ : \ 
     \exists y\ ( \; 
       Sa_i \uund Ea_iy \uund Ty
     \; ) 
  \ \}\,.
\]
Therefore, by Theorem~\ref{thm:UCQtesting}, $q'$ and hence $\qSET$ can be answered with constant
update time and constant answer time on all databases that satisfy $\SD$.
\end{example}

For handling the general case, assume we are given a
set $\CONSTR$ of small domain constraints and
an arbitrary $k$-ary CQ $q$
of the form \eqref{eq:karyCQ} where $\phi$ is of the form \eqref{eq:CF}.
We define a function $\textit{Dom}_{q,\CONSTR}$ that maps each $x\in\Vars(q)$ to a set 
$\textit{Dom}_{q,\CONSTR}(x)\subseteq\Dom$ as follows. 
As an initialisation let $f(x)=\Dom$ for each $x\in\Vars(q)$. Consider each constraint $\DEP$ in $\CONSTR$ and
let $S[i]\subseteq C$ be the form of $\DEP$. Consider each atom $\psi_j$ of $\phi$ and let 
$Rv_1\cdots v_r$ be the form of $\psi_j$. If $R=S$ and $v_i\in\Var$, then let
$f(v_i)\deff f(v_i)\cap C$. Let $\textit{Dom}_{q,\CONSTR}$ be the mapping $f$ obtained at the end of this process.
Note that $\textit{rvars}_\CONSTR(q) \deff \setc{x\in\Vars(q)}{\textit{Dom}_{q,\CONSTR}(x)\neq\Dom}$
consists of the variables of $q$ that are restricted by $\CONSTR$. 

Let $M_{q,\CONSTR}$ be the set of all mappings $\alpha:V\to\Dom$ with $V=\textit{rvars}_\CONSTR(q)$
and $\alpha(x)\in \textit{Dom}_{q,\CONSTR}(x)$ for each $x\in V$. Note that $M_{q,\CONSTR}$ is finite; and it is 
empty if, and only if, $\textit{Dom}_{q,\CONSTR}(x)=\emptyset$ for some $x\in \Vars(q)$.

For an arbitrary mapping $\alpha:V\to \Dom$ with $V\subseteq \Var$ we let 
$q_\alpha$ be the $k$-ary CQ obtained from $q$ as follows: 
for each $x\in V$, if present in $q$, the existential quantifier 
``$\exists x$'' is omitted, and afterwards
every occurrence of $x$ in $q$ is replaced with the constant
$\alpha(x)$.
It is straightforward to check that $q_\alpha(\DB)\subseteq q(\DB)$
for every $\schema$-db $\DB$.
With these notations, we obtain the following lemma.

\begin{lemma}\label{lemma:SD-constraints}
Let $q$ be a CQ and let $\CONSTR$ be a set of small domain constraints.
Let $M\deff M_{q,\CONSTR}$.
\\
If $M=\emptyset$, then $q(\DB)=\emptyset$ for every $(\schema,\CONSTR)$-db 
$\DB$.
\\
Otherwise, $q$ is $\CONSTR$-equivalent to the UCQ 
\ $
  q_{\CONSTR} \ \deff \ 
  \bigcup_{\alpha\in M}\, q_\alpha 
$\,.
\end{lemma}

\begin{proof}
For a set $\CONSTR$ of constraints and
a $\schema$-db $\DB$ we write $\DB\models\CONSTR$ to indicate that  
$\DB$ satisfies every constraint in $\CONSTR$.

\smallskip

\noindent
Let $V\deff\textit{rvars}_{\CONSTR}(q)$ and $M\deff M_{q,\CONSTR}$.
If $V=\emptyset$, then $M=\set{\alpha_\emptyset}$ where
$\alpha_\emptyset$ is the unique mapping with empty domain.
Thus, $q_\CONSTR=q_{\alpha_\emptyset}=q$ and we are done.
It remains to consider the case where $V\neq\emptyset$.

Consider an arbitrary $\schema$-db $\DB$ with $\DB\models\CONSTR$ and
let $q$ be of the form \eqref{eq:karyCQ}.
Consider an arbitrary tuple $\ov{b}=(b_1,\ldots,b_k)\in q(\DB)$. By definition of the
semantics of CQs, there is a
valuation $\beta:\Var\to\Dom$ such that
$(b_1,\ldots,b_k)=\big(\beta(u_1),\ldots,\beta(u_k)\big)$ and
for every atomic formula $Rv_1\cdots v_r$ in $q$ we have
$\big(\beta(v_1),\ldots,\beta(v_r)\big)\in R^\DB$.
If $x=v_i$ then $\beta(x)\in\pi_i(R^\DB)$; and if $\CONSTR$ contains a
constraint of the form $R[i]\subseteq C$ then, since
$\DB\models\CONSTR$, we have $\pi_i(R^\DB)\subseteq C$, and hence
$\beta(x)\in C$. This holds true for every occurrence of $x$ in
an atom of $q$, and hence $\beta(x)\in \textit{Dom}_{q,\CONSTR}(x)$
for every $x\in\Vars(q)$. In other words, the restriction $\beta_{|V}$
of $\beta$ to $V$ belongs to $M$, and
$\ov{b}\in q_{\beta_{|V}}(\DB)\subseteq
q_{\CONSTR}(\DB)$.
In particular, this implies that the following is true.
\begin{enumerate}
\item
If $q(\DB)\neq\emptyset$ for some $\schema$-db $\DB$ with
$\DB\models\CONSTR$, then $M\neq \emptyset$. Hence, by contraposition,
if $M=\emptyset$ then $q(\DB)=\emptyset$ for every $\schema$-db $\DB$
with $\DB\models\CONSTR$.
\item
If $M\neq\emptyset$, then $q(\DB)\subseteq q_{\CONSTR}(\DB)$ for every
$\schema$-db $\DB$ with $\DB\models\CONSTR$.
On the other hand,
since $q_\alpha(\DB)\subseteq q(\DB)$ for every $\alpha$ and every
$\schema$-db $\DB$, we have $q_\CONSTR(\DB)\subseteq q(\DB)$ for
every $\schema$-db $\DB$. Hence, $q$ is $\CONSTR$-equivalent to $q_\CONSTR$.
\end{enumerate}

\noindent
This completes the proof of Lemma~\ref{lemma:SD-constraints}.
\end{proof}

This reduction from a CQ $q$ to a UCQ $q_{\CONSTR}$ directly
translates to UCQs: if $q$ is a union of the CQs
$q_1,\ldots, q_d$, then we define the UCQ
$q_{\CONSTR} \deff \bigcup_{i\in [d]}\, (q_i)_{\CONSTR} $.  It is not
hard to verify that if the UCQ $q$ is a homomorphic core, then so is
$q_{\CONSTR}$. Therefore, the following dichotomy theorem for UCQs under small
domain constraints is a direct consequence of
Lemma~\ref{lemma:SD-constraints} and the
Theorems~\ref{thm:UCQtesting}, \ref{thm:enumUCQ}, and \ref{thm:countUCQ}.

\begin{theorem}\label{thm:SD-constraints}
Let $q$ be a UCQ that is a homomorphic core 
and $\CONSTR$ a set of small domain constraints 
with $M_{q,\CONSTR}\neq \emptyset$. 
Suppose that the \OMvcon and the \OVcon hold.
\begin{enumerate}[(1a)]
\item[(1a)] If $q_{\CONSTR}$ is t-hierarchical, then $q$ can be tested 
on $(\schema,\CONSTR)$-dbs
in constant time
with linear preprocessing time and constant update time.
\item[(1b)] If $q_{\CONSTR}$ is not t-hierarchical, then 
on the class of $(\schema,\CONSTR)$-dbs
testing in time $O(n^{1-\epsilon})$ is not possible 
with $O(n^{1-\epsilon})$ update time.
\item[(2a)] If $q_{\CONSTR}$ is \qhier, then there is data structure with linear
preprocessing and constant update time that allows to enumerate $q(D)$ with constant
delay
on $(\schema,\CONSTR)$-dbs.
\item[(2b)] If $q_{\CONSTR}$ is not \qhier and in addition self-join free, then
 $q(D)$ cannot be enumerated with
$O(n^{1-\epsilon})$ delay and $O(n^{1-\epsilon})$ update time
on $(\schema,\CONSTR)$-dbs.
\item[(3a)] If $q_{\CONSTR}$ is exhaustively \qhier, then there is data structure with linear
preprocessing and constant update time that allows to compute $|q(D)|$ in constant
time
on $(\schema,\CONSTR)$-dbs.
\item[(3b)] If $q_{\CONSTR}$ is not exhaustively \qhier, then computing
  $|q(D)|$ 
on $(\schema,\CONSTR)$-dbs
in time
  $O(n^{1-\epsilon})$ is not possible
with $O(n^{1-\epsilon})$ update time.
\end{enumerate}
\end{theorem}

In particular, this shows that the tractability of a UCQ $q$ on
$(\schema,\CONSTR)$-dbs only depends on the 
structure of the query $q_{\CONSTR}$. Note that while the size of
$q_{\CONSTR}$ might be $c^{O(q)}$, where $c$ is largest number of
constants in a small domain, it can be checked in time $\poly(q)$
whether $q_{\CONSTR}$ is t-hierarchical or q-hierarchical.

\smallskip

Let us take a brief look at two other kinds of constraints: inclusion
dependencies and functional dependencies, which both can also cause a hard
query to become tractable. 

An \emph{inclusion dependency $\DEP$} is of the form
\;$
  R[i_1,\ldots,i_m]\subseteq S[j_1,\ldots,j_m]
$\;
where $R,S\in\schema$, $m\geq 1$, $i_1,\ldots,i_m\in\set{1,\ldots,\ar(R)}$, and $j_1,\ldots,j_m\in\set{1,\ldots,\ar(S)}$.
A $\schema$-db $\DB$ \emph{satisfies} $\DEP$ if 
$\proj_{i_1,\ldots,i_m}(R^{\DB}) \subseteq \proj_{j_1,\ldots,j_m}(S^{\DB})$.
As an example consider the query $\qSET$
from Example~\ref{example:SD-constraints} and the inclusion
dependency $\IND \ \deff\  E[2]\subseteq T[1]$. Obviously, $\qSET$ is 
$\set{\IND}$-equivalent to the \qhier (and hence easy) CQ
$
 q' \deff \{ \ 
 \emptytuple \ : \ 
 \exists x\exists y\ (\; Sx \uund Exy \;)
 \ \}
 $.
 To turn this into a
general principle, we say that an inclusion dependency $\DEP$
of the form
$R[i_1,\ldots,i_m]\subseteq S[j_1,\ldots,j_m]$
\emph{can be applied} to a CQ $q$
if $q$ contains an atom $\psi_1$ of the form $Rv_1\cdots v_{r}$ and an
atom $\psi_2$ of the form $Sw_1\cdots w_{s}$ such that
\begin{enumerate}
\item
$(v_{i_1},\ldots,v_{i_m})=(w_{j_1},\ldots,w_{j_m})$, 
\item
for all
$j\in[s]\setminus\set{j_1,\ldots,j_m}$ we
have
$w_j\in\Var$,
$w_j\not\in\free(q)$, $\atoms(w_j)=\set{\psi_2}$, and
\item
for all $j,j'\in[s]\setminus\set{j_1,\ldots,j_m}$ with $j\neq j'$ we
have $w_j\neq w_{j'}$;
\end{enumerate} 

\noindent
and \emph{applying
$\DEP$ to $q$ at $(\psi_1,\psi_2)$} then yields the CQ $q'$ which is obtained from $q$ by
omitting the atom $\psi_2$ and omitting the quantifiers $\exists z$
for all $z\in\Vars(\psi_2)\setminus\set{w_{j_1},\ldots,w_{j_m}}$. 
By this construction we have 
$\Vars(q')=\Vars(q) \setminus \setc{w_j}{j\in [s]\setminus\set{j_1,\ldots,j_m}}$.

\begin{claim}\label{claim:indequiv}
$q'\equiv_{\set{\DEP}}q$, and if $q$ is q-hierarchical, then so is
$q'$. 
\end{claim}

\begin{proof}
For a set $\CONSTR$ of constraints and
a $\schema$-db $\DB$ we write $\DB\models\CONSTR$ to indicate that  
$\DB$ satisfies every constraint in $\CONSTR$.
For a constraint $\DEP$
we write $\DB\models\DEP$ instead of $\DB\models\set{\DEP}$.

\smallskip

\noindent
Obviously, $q(\DB)\subseteq q'(\DB)$ for every $\schema$-db $\DB$.
For the opposite direction, let $q$ be of the form \eqref{eq:karyCQ}, 
and consider a $\schema$-db $\DB$ with
$\DB\models\IND$ and a tuple $t\in
q'(\DB)$. Our goal is to show that $t\in q(\DB)$.
Since $t\in q'(\DB)$, there is a valuation $\beta'$ such that 
$t=\big(\beta'(u_1),\ldots,\beta'(u_k)\big)$ and
$(\DB,\beta')\models\psi$ for each atom $\psi$ of $q'$. In particular, 
$(\DB,\beta')\models Rv_1\cdots v_{r}$, i.e.,
$\big(\beta'(v_1),\ldots,\beta'(v_{r})\big)\in R^\DB$.
To show that $t\in q(\DB)$ it suffices to modify $\beta'$ into a
valuation $\beta$ which coincides with $\beta'$ on all variables in
$\Vars(q')$ and which also ensures that
$(\DB,\beta)\models Sw_1\cdots w_{s}$, i.e., that
$\big(\beta(w_1),\ldots,\beta(w_{s})\big)\in S^\DB$.

Since $\DB\models\IND$ we obtain from $\big(\beta'(v_1),\ldots,\beta'(v_{r})\big)\in R^\DB$
that $\big(\beta'(v_{i_1}),\ldots,\beta'(v_{i_m})\big)\in
\pi_{i_1,\ldots,i_m}(R^\DB)\subseteq \pi_{j_1,\ldots,j_m}(S^\DB)$. 
Since $(v_{i_1},\ldots,v_{i_m})=(w_{j_1},\ldots,w_{j_m})$, this
implies that
$\big(\beta'(w_{j_1}),\allowbreak \ldots,\allowbreak
\pi_{j_1,\ldots,j_m}(S^\DB)$.
Hence, there exists a tuple $(a_1,\ldots,a_s)\in S^\DB$ such that
$\big(\beta'(w_{j_1}),\allowbreak\ldots,\allowbreak \beta'(w_{j_m})\big) =
(a_{j_1},\ldots,a_{j_m})$.

We let $\beta$ be the valuation obtained from $\beta'$ by letting
$\beta(w_j)\deff a_j$ for every $j\in
[s]\setminus\set{j_1,\ldots,j_m}$.
With this choice we have $(\DB,\beta)\models Sw_1\cdots w_s$.
Note that $\beta$ differs from $\beta'$ only in variables $w_j$ for
which we know that
$\atoms(w_j)=\set{\psi_2}=\set{Sw_1\cdots w_s}$, i.e., variables
that occur in no other atom of $q$ than the atom $Sw_1\cdots w_s$.
Therefore, $(\DB,\beta)\models\psi$ for each atom $\psi$ of $q$, and
hence
$t=\big(\beta(u_1),\ldots,\beta(u_k)\big)\in q(\DB)$.

In summary, we obtain that $q'(\DB)\subseteq q(\DB)$ for every
$\schema$-db $\DB$ with $\DB\models \IND$. This completes the
proof showing that $q'\equiv_{\set{\IND}} q$.

To verify the claim's second statement, 
let $W\deff\setc{w_j}{j\in [s]\setminus\set{j_1,\ldots,j_m}}$ and
note that $\Vars(q')=\Vars(q)\setminus W$. 
For all $x\in W$ we have $\atoms_q(x)=\set{\psi_2}$,  
and for all $x\in\Vars(q')$ we have
$\atoms_{q'}(x)=\atoms_{q}(x)\setminus\set{\psi_2}$.
Using this, we obtain that if $q$ satisfies
condition \eqref{eq:q-hier-cond-i} of
Definition~\ref{def:qhierarchical} then so does $q'$. 

It remains to show that if $q$ is
q-hierarchical, then $q'$ also satisfies condition \eqref{eq:q-hier-cond-ii} of
Definition~\ref{def:qhierarchical}. 
Assume for contradiction that $q'$ does not satisfy this
condition. Then, there are $x\in\free(q')$ and
$y\in\Vars(q')\setminus\free(q')$ with
$\atoms_{q'}(x)\varsubsetneq\atoms_{q'}(y)$.

\emph{Case~1:} $\atoms_{q}(x)=\atoms_{q'}(x)$.
Then, $\atoms_q(x)\varsubsetneq\atoms_{q}(y)$, and thus $q$ does not
satisfy condition \eqref{eq:q-hier-cond-ii} of
Definition~\ref{def:qhierarchical} and hence is not q-hierarchical.

\emph{Case~2:} $\atoms_{q}(x)=\atoms_{q'}(x)\cup\set{\psi_2}$.
If $\atoms_q(y)=\atoms{q'}(y)\cup\set{\psi_2}$, then we are done by the
same reasoning as in Case~1.
On the other hand, if $\atoms_q(y)=\atoms_{q'}(y)$, then
$\psi_2\in\atoms_{q}(x)\setminus\atoms_q(y)$.
Furthermore, since $\atoms_{q'}(x)\varsubsetneq\atoms_{q'}(y)$, there are atoms
$\psi$ and $\psi'$ such that
$\psi\in\atoms_{q}(x)\cap\atoms_q(y)$ and
$\psi'\in\atoms_q(y)\setminus\atoms_q(x)$.
Thus, $q$ violates condition \eqref{eq:q-hier-cond-i} of
Definition~\ref{def:qhierarchical} and hence is not q-hierarchical.
\end{proof}

From the claim it follows that we can simplify a given query by
iteratively applying inclusion dependencies to pairs of atoms of the query.
In some cases, this transforms queries
that are hard in general into $\CONSTR$-equivalent queries that are
q-hierarchical and hence easy for dynamic evaluation.
For example, an iterated application of $\IND\deff E[2]\subseteq E[1]$
transforms the non-t-hierarchical query
\ $
 \{\, 
  (x,y) \, :\,
  \exists z_1\exists z_2\;(\,Exy \uund Eyz_1 \uund Ez_1z_2\,)
 \,\}
$ \
into the q-hierarchical query 
\ $
 \{\,
  (x,y) \, : \,
  Exy
  \,\}.
$
However, the limitations of this approach are documented by the query
\ $
 q\deff \{\,
  \emptytuple \,:\,
  \exists x\exists y\exists z\exists z'\,(\,
    Sx \uund Exy \uund Ty \uund Rzz'
  \,)
 \,\},
$ 
which is $\CONSTR$-equivalent to the \qhier query 
\ $
 q'\deff \{\,
  \emptytuple \,:\,
  \exists z\exists z'\; Rzz'
 \,\},
$
for $\CONSTR\deff\set{\, R[1,2]\subseteq E[1,2]\;,\; R[1]\subseteq
  S[1]\;,\; R[2]\subseteq T[1]\,}$, but where $q'$ cannot be obtained
by iteratively applying dependencies of $\CONSTR$ to $q$.

\smallskip

The presence of
\emph{functional dependencies} can also cause a hard query to become
tractable: 
Consider the functional dependency
$
 \FD \deff 
 E[1\to 2]
 $,
which is satisfied by a database $\DB$ iff for every $a\in\Dom$ there
is at most one $b\in\Dom$ such that $(a,b)\in E^\DB$.
On databases that satisfy $\FD$, 
the query $\qSET$ from Example~\ref{example:SD-constraints} 
can be evaluated with constant
 answer time and constant update time as follows:
 One can
 store for every $b$ the number $m_b$ of elements $(a,b)\in E^\DB$ such that
 $a\in S^\DB$ and in addition the number $m=\sum_{b\in T^\DB}m_b$,
 which is non-zero if and only if 
 $\qSET(\DB)=\Yes$.
 The functional
 dependency guarantees that every update affects  at most one number
 $m_b$ and one summand of
 $m$. Using constant access data structures, the query result can
 therefore be maintained with constant update time.

 The nature of this example is somewhat different compared to the
 approaches for small
 domain constraints or inclusion constraints described above: We can show that
 the query becomes tractable, but we are not aware of any $\{\FD\}$-equivalent
 q-hierarchical CQ or UCQ that would explain its tractability via a
 reduction to the setting without integrity constraints.
 To exploit the full power of functional dependencies for improving
 dynamic query evaluation, it seems therefore necessary to come up with
 new algorithmic approaches that go beyond the techniques we have for
 (q- or t-)hierarchical queries.

\bibliographystyle{plainurl}
\bibliography{literature}

\clearpage
\appendix

\section*{APPENDIX}
\section{Full proof of Theorem~\ref{thm:testing}\,\eqref{item:thm:testing:lowerbound} }
\label{appendix:TestingLowerBound}

\newcommand{\ellalt}{{\ell'}}
\newcommand{\phione}{\phiBTypical}
\newcommand{\phitwo}{\phiET}
\newcommand{\phiBool}{\queryphi_\exists}
\newcommand{\domn}{\ensuremath{\textup{dom}_\dimn}}

\subsection*{Proof of Theorem~\ref{thm:testing}\,\eqref{item:thm:testing:lowerbound} for the case that
$q$ violates condition \eqref{eq:t-hier-cond-i} 
of Definition~\ref{def:thierarchical}}

Assume we are given a query $q \deff q_\phi(z_1,\ldots,z_k)$ that is a homomorphic
core and that is not t-hierarchical because it 
violates condition \eqref{eq:t-hier-cond-i} 
of Definition~\ref{def:thierarchical}.
Thus, there
are two variables $\varx, \vary\in\Vars(q)\setminus\free(q)=\Vars(q)\setminus\set{z_1,\ldots,z_k}$ and three atoms
$\sgpsix, \sgpsixy, \sgpsiy$ of $q$ 
with
$\Vars(\sgpsix)\cap\set{\varx,\vary}=\set{\varx}$,
$\Vars(\sgpsixy)\cap\set{\varx,\vary}=\set{\varx,\vary}$, and
$\Vars(\sgpsiy)\cap\set{\varx,\vary}=\set{\vary}$.  

We show how a dynamic algorithm that solves the testing problem for $q$ can be used to
solve the OuMv-problem. 

Without loss of
generality we assume that
$\Vars(q)=\set{\varx,\vary,\varz_1,\ldots,\varz_\ell}$ for some $\ell\geq k$, and 
$|\Vars(q)|=\ell+2$.
For a
given $\dimn\times\dimn$ matrix $\matM$
we fix a domain $\domn$ that consists of $2n+\ell$ elements  
$\setc{\verta_\indi, \vertb_\indi}{\indi\in[\dimn]} \cup \setc{\vertc_\inds}{\inds\in[\ell]}$ 
from $\Dom\setminus\Cons(q)$.
For $\indi,\indj\in[\dimn]$ we let
$\iotasubij$
be the injective mapping
from $\Vars(q)\cup\Cons(q)$ to $\domn\cup\Cons(q)$
with
\begin{itemize}
\item
$\iotasubij(\varx)=\verta_\indi$, 
\item
$\iotasubij(\vary)=\vertb_\indj$,
\item
$\iotasubij(\varz_\inds)=\vertc_\inds$ \ for all $s\in[\ell]$, \ and
\item
$\iotasubij(d)=d$ \ for all $d\in\Cons(q)$.
\end{itemize}
We tacitly extend $\iotasubij$ to a mapping from $\Vars(q)\cup\Dom$ to $\Dom$ by letting
$\iotasubij(d)=d$ for every $d\in\Dom$.

For the matrix $\matM$ and for
$\dimn$-dimensional vectors $\vecu$ and $\vecv$, 
we define a $\sigma$-db
$\DB=\DB(q,\matM,\vecu,\vecv)$ with 
$\adom{\DB}\subseteq\domn\cup\Cons(q)$ as follows (recall our notational
convention that $\vecu_\indi$ denotes the $\indi$-th entry of a vector
$\vecu$).
For every atom
$\sgpsi=\relR\varw_1\cdots\varw_\arityr$ in $q$ we include
in $\relR^\DB$ the tuple
$\big(\iotasubij(\varw_1),\ldots,\iotasubij(\varw_\arityr)\big)$
\begin{itemize}
\item
for all $\indi,\indj\in[\dimn]$ such that 
 $\vecu_\indi=1$, \
 if \,$\sgpsi=\sgpsix$,  
\item
for all $\indi,\indj\in[\dimn]$ such that 
$\vecv_\indj=1$, \ 
if \,$\sgpsi=\sgpsiy$, 
\item
for all $\indi,\indj\in[\dimn]$ such that $\matM_{\indi,\indj}=1$, \ if
\,$\sgpsi=\sgpsixy$, \ and
\item 
for all $\indi,\indj\in[\dimn]$, \ if \,$\sgpsi\notin\set{\,\sgpsix,\,\sgpsixy,\,\sgpsiy\,}$. 
\end{itemize}

Note that the relations in the atoms $\sgpsix$, $\sgpsiy$, and
$\sgpsixy$ are used to encode $\vecu$, $\vecv$, and $\matM$,
respectively.  
Moreover, 
since $\sgpsix$ ($\sgpsiy$) does not contain the variable $\vary$ ($\varx$),
two databases $\DB=\DB(q,\matM,\vecu,\vecv)$ and
$\DB'=\DB(q,\matM,\mbox{$\vecu\,{}'$},\mbox{$\vecv\,{}'$})$ differ only in at most
$2\dimn$ tuples.
Therefore, $\DB'$ can be obtained from $\DB$ by $2\dimn$ update steps. 
It follows from the definitions that $\iotasubij$ is a 
homomorphism from $q$ to $\DB$ if and only if $\vecu_\indi=1$,
$\vecv_\indj=1$, and $\matM_{\indi,\indj}=1$.  Therefore,
$\vecu\trans \matM \vecv = 1$ if and only if there are
$\indi,\indj\in[n]$ such that $\iotasubij$ is a
homomorphism from $q$ to $\DB$.

We let 
$\homDBtoquery$ be the (surjective) mapping from $\domn\cup\Cons(q)$ to
$\Vars(q)\cup\Cons(q)$ 
defined by 
$\homDBtoquery(d)=d$ for all $d\in\Cons(q)$ and
$\homDBtoquery(\vertc_\inds)\defi\varz_\inds$,
$\homDBtoquery(\verta_\indi)\defi\varx$,
$\homDBtoquery(\vertb_\indj)\defi\vary$ for all $\indi,\indj\in[n]$
and $\inds\in[\ell]$.
Clearly, $\homDBtoquery$ is a
homomorphism from $\DB$ to $q$.
Obviously, the following is true for every mapping
$\homh$ from $\Vars(q)$ to $\Adom(\DB)$ and for all $\varw\in\Vars(q)$:
\begin{itemize}
\item
  if $\homh(w)=\vertc_\inds$ for some $\inds\in[\ell]$, 
  then $(\homDBtoquery\circ h)(w)=\varz_\inds$,
\item
  if $\homh(w)=\verta_\indi$ for some $\indi\in[n]$, 
  then $(\homDBtoquery\circ h)(w)=\varx$,
\item
  if $\homh(w)=\vertb_\indj$ for some $\indj\in[n]$, 
  then $(\homDBtoquery\circ h)(w)=\vary$,
\item
  if $\homh(w)=d$ for some $d\in\Cons(q)$, 
  then $(\homDBtoquery\circ h)(w)=d$.
\end{itemize}

We define the partition
$\partPfull=\Set{\set{\vertc_1},\ldots,\set{\vertc_\ell},\setc{\verta_\indi}{\indi\in[\dimn]},\setc{\vertb_\indj}{\indj\in[\dimn]}}$
of $\domn$ and say that a mapping
$\homh$ from $\Vars(q)\cup\Dom$ to $\Dom$
\emph{respects} $\partPfull$,
if for each set from the partition there is exactly one element in the
image of $\Vars(q)$ under $\homh$, i.e., the set 
$\setc{h(w)}{w\in\Vars(q)}$.
\begin{claim}\label{claim:respectinghom_uMv}
  $\vecu\trans \matM \vecv = 1$ \ $\iff$ \
  There exists a homomorphism $\homh\colon q\to \DB$ that
  respects $\partPfull$.
\end{claim}
\begin{proof}
  For one direction assume that $\vecu\trans \matM \vecv = 1$.  Then
  there are $\indi,\indj\in[\dimn]$ such that $\iotasubij$ is a homomorphism
  from $q$ to $\DB$
  that respects $\partPfull$.  For the other direction assume that
  $\homh\colon q\to \DB$ is a homomorphism that respects
  $\partPfull$.  
  Thus, there are elements $w_{a}, w_b, w_1,\ldots,w_{\ell}$ in $\Vars(q)$ such that
  $\homh(w_{a})\in\setc{a_i}{i\in[n]}$, $\homh(w_b)\in\setc{b_j}{j\in[n]}$, and $\homh(w_s)=c_s$ for each $s\in[\ell]$.
It follows that $(\homDBtoquery\circ\homh)$ is a
  bijective homomorphism from $q_\phi(z_1,\ldots,z_k)$ to $q_\phi((\homDBtoquery\circ\homh)(z_1),\ldots,(\homDBtoquery\circ\homh)(z_k))$.
  Therefore, 
  it can easily be verified that
  $\homh\circ(\homDBtoquery\circ\homh)^{-1}$ is a homomorphism from
  $q$ to $\DB$ which equals $\iotasubij$ for some
  $\indi,\indj\in[\dimn]$.  This implies that $\vecu\trans \matM \vecv = 1$.
\end{proof}
\begin{claim}\label{claim:respectinghom_core}
If  $q$ is a homomorphic core, then every homomorphism $\homh\colon q\to \DB$ respects $\partPfull$. 
\end{claim}
\begin{proof}
  Assume for contradiction that $\homh\colon q\to \DB$ is a
  homomorphism that does not respect $\partPfull$.  Then
  $(\homDBtoquery\circ\homh)$ is a homomorphism from $q$ into a
  proper subquery of $q$, contradicting that $q$
  is a homomorphic core.
\end{proof}

\begin{claim}\label{claim:testing-lowerbound-i}
If $q$ is a homomorphic core, then \
$\vecu\trans \matM \vecv = 1$
 $\iff$ 
$(c_1,\ldots,c_k)\in q(\DB)$\,.
\end{claim}
\begin{proof}
We already know that $\vecu\trans \matM \vecv = 1$ if and only if there are $i,j\in[n]$ such that 
$\iotasubij$ is a homomorphism from $q$ to $\DB$.
Furthermore, $\iotasubij(z_s)=c_s$ for all $s\in[\ell]$ and all $i,j\in[n]$.
Thus, if 
$\vecu\trans\matM \vecv = 1$, then there exist $i,j\in[n]$ such that 
$\big(\iotasubij(z_1),\ldots,\iotasubij(z_k)\big)=(c_1,\ldots,c_k)\in q(\DB)$.
This proves direction ``$\Longrightarrow$'' of the claim.
For the opposite direction, note that if $(c_1,\ldots,c_k)\in q(\DB)$, then there is
a homomorphism $h$ from $q$ to $\DB$. By Claim~\ref{claim:respectinghom_core}, $h$ respects
$\partPfull$, and hence by Claim~\ref{claim:respectinghom_uMv}, $\vecu\trans \matM \vecv = 1$.
\end{proof}

We are now ready for proving Theorem~\ref{thm:testing}\,\eqref{item:thm:testing:lowerbound} for the case that
$q$ violates condition \eqref{eq:t-hier-cond-i} 
of Definition~\ref{def:thierarchical}.
  Assume for contradiction that the
  testing problem
  for $q$ can be solved with update time
  $\updatetime=\bigoh(\actdomsize^{1-\smalleps})$ 
  and testing time $\testingtime=\bigoh(\actdomsize^{2-\smalleps})$.
  We can use this
  algorithm to solve the OuMv-problem in time
  $\bigoh(\actdomsize^{3-\smalleps})$ as follows.

  In the
  preprocessing phase, we are given the $\dimn\times\dimn$ matrix
  $\matM$ and let $\vecu^{\,0}$, $\vecv^{\,0}$ be the all-zero vectors
  of dimension $\dimn$.
  We start the preprocessing phase of our testing algorithm for
  $q$ with the empty database. As this database has constant size,
  the preprocessing phase finishes in constant time.
  Afterwards, we use $\bigOh(\dimn^2)$ insert operations to build the
  database $\DB(q,\matM,\vecu^{\,0},\vecv^{\,0})$.
  All this is done within time $\bigOh(n^2\updatetime)=\bigOh(n^{3-\epsilon})$.

  When a pair
  of vectors $\vecu^{\,\indt}$, $\vecv^{\,\indt}$ (for $t\in[\dimn]$) arrives, we change the
  current database \allowbreak
  $\DB(q,\matM,\vecu^{\,\indt-1},\vecv^{\,\indt-1})$ into 
  $\DB(q,\matM,\vecu^{\,\indt},\vecv^{\,\indt})$ by using
  at most $2\dimn$ update steps.  
  By Claim~\ref{claim:testing-lowerbound-i}
  we know that
  $(\vecu^{\,\indt})\trans \matM \vecv^{\,\indt} = 1$ if, and only if, 
  $(c_1,\ldots,c_k)\in q(\DB)$, for
  $\DB \deff \DB(q,\matM,\vecu^{\,\indt},\vecv^{\,\indt})$.
  Hence, after running the $\TEST$ routine with input $(c_1,\ldots,c_k)$
  in time $\testingtime=\bigOh(\setsize{\adom{\DB}}^{2-\smalleps}) \allowbreak = \bigoh(\dimn^{2-\smalleps})$
  we can output the value of $(\vecu^{\,\indt})\trans \matM \vecv^{\,\indt}$.

  The time we spend for each $t\in[n]$ is bounded by
  $\bigOh(2\dimn\updatetime+\testingtime)=\bigoh(\dimn^{2-\smalleps})$.
  Thus, the
  overall running time for solving the OuMv-problem sums up to $\bigoh(\dimn^{3-\smalleps})$, contradicting 
  the OuMv-conjecture and hence also the OMv-conjecture.

This completes the proof of Theorem~\ref{thm:testing}\,\eqref{item:thm:testing:lowerbound} for the case that
$q$ violates condition \eqref{eq:t-hier-cond-i} 
of Definition~\ref{def:thierarchical}.
\qed

\subsection*{Proof of Theorem~\ref{thm:testing}\,\eqref{item:thm:testing:lowerbound} for the case that
$q$ violates condition \eqref{eq:t-hier-cond-ii} 
of Definition~\ref{def:thierarchical}}

Assume we are given a query $q \deff q_\phi(z_1,\ldots,z_k)$ that is a homomorphic
core and that is not t-hierarchical because it 
violates condition \eqref{eq:t-hier-cond-ii} 
of Definition~\ref{def:thierarchical}.
Thus, there
are two variables $\varx\in\free(q)$ and $\vary\in\Vars(q)\setminus\free(q)$ and two atoms
$\sgpsixy$ and $\sgpsiy$ of $q$ 
with
$\Vars(\sgpsixy)\cap\set{\varx,\vary}=\set{\varx,\vary}$ and
$\Vars(\sgpsiy)\cap\set{\varx,\vary}=\set{\vary}$.  

We show how a dynamic algorithm that solves the testing problem for $q$ can be used to
solve the OuMv-problem. 

Without loss of
generality we assume that 
$\Vars(q)=\set{z_1,\ldots,z_\ell}$ with 
$\free(q)=\set{z_1,\ldots,z_k}$, $\ell>k$, 
$x=z_1$, and $y=z_{\ell}$.

For a given $\dimn\times\dimn$ matrix $\matM$
we fix a domain $\domn$ that consists of $2n+\ell-2$ elements  
$\setc{\verta_\indi, \vertb_\indi}{\indi\in[\dimn]} \cup \setc{\vertc_\inds}{\inds\in\set{2,\ldots,\ell{-}1}}$ 
from $\Dom\setminus\Cons(q)$.
For $\indi,\indj\in[\dimn]$ we let
$\iotasubij$
be the injective mapping
from $\Vars(q)\cup\Cons(q)$ to $\domn\cup\Cons(q)$
with
\begin{itemize}
\item
$\iotasubij(\varx)=\verta_\indi$, 
\item
$\iotasubij(\vary)=\vertb_\indj$,
\item
$\iotasubij(\varz_\inds)=\vertc_\inds$ \ for all $s\in\set{2,\ldots,\ell{-}1}$, \ and
\item
$\iotasubij(d)=d$ \ for all $d\in\Cons(q)$.
\end{itemize}
We tacitly extend $\iotasubij$ to a mapping from $\Vars(q)\cup\Dom$ to $\Dom$ by letting
$\iotasubij(d)=d$ for every $d\in\Dom$.

For the matrix $\matM$ and for
an $\dimn$-dimensional vector $\vecv$, 
we define a $\schema$-db
$\DB=\DB(q,\matM,\vecv)$ with 
$\adom{\DB}\subseteq\domn\cup\Cons(q)$ as follows (recall our notational
convention that $\vecv_\indj$ denotes the $\indj$-th entry of a vector
$\vecv$).
For every atom
$\sgpsi=\relR\varw_1\cdots\varw_\arityr$ in $q$ we include
in $\relR^\DB$ the tuple
$\big(\iotasubij(\varw_1),\ldots,\iotasubij(\varw_\arityr)\big)$
\begin{itemize}
\item
for all $\indi,\indj\in[\dimn]$ such that 
$\vecv_\indj=1$, \ 
if \,$\sgpsi=\sgpsiy$, 
\item
for all $\indi,\indj\in[\dimn]$ such that $\matM_{\indi,\indj}=1$, \ if
\,$\sgpsi=\sgpsixy$, \ and
\item 
for all $\indi,\indj\in[\dimn]$, \ if \,$\sgpsi\notin\set{\,\sgpsixy,\,\sgpsiy\,}$. 
\end{itemize}

Note that the relations in the atoms $\sgpsiy$ and
$\sgpsixy$ are used to encode $\vecv$ and $\matM$,
respectively.  
Moreover, 
since $\sgpsiy$ does not contain the variable $\varx$,
two databases $\DB=\DB(q,\matM,\vecv)$ and
$\DB'=\DB(q,\matM,\mbox{$\vecv\,{}'$})$ differ only in at most
$\dimn$ tuples.
Therefore, $\DB'$ can be obtained from $\DB$ by $\dimn$ update steps. 
It follows from the definitions that 
\[
\text{$\iotasubij$ is a 
homomorphism from $q$ to $\DB$}
\ \ \iff \ \ 
\text{$\matM_{\indi,\indj}=1$ and $\vecv_\indj=1$.}
\]  

We let 
$\homDBtoquery$ be the (surjective) mapping from $\domn\cup\Cons(q)$ to
$\Vars(q)\cup\Cons(q)$ 
defined by 
$\homDBtoquery(d)=d$ for all $d\in\Cons(q)$ and
$\homDBtoquery(\vertc_\inds)\defi\varz_\inds$,
$\homDBtoquery(\verta_\indi)\defi\varx$,
$\homDBtoquery(\vertb_\indj)\defi\vary$ for all $\indi,\indj\in[n]$
and $\inds\in\set{2,\ldots,\ell{-}1}$.
Clearly, $\homDBtoquery$ is a
homomorphism from $\DB$ to $q$.
Obviously, the following is true for every mapping
$\homh$ from $\Vars(q)$ to $\Adom(\DB)$ and for all $\varw\in\Vars(q)$:
\begin{itemize}
\item
  if $\homh(w)=\vertc_\inds$ for some $\inds\in\set{2,\ldots,\ell{-}1}$, 
  then $(\homDBtoquery\circ h)(w)=\varz_\inds$,
\item
  if $\homh(w)=\verta_\indi$ for some $\indi\in[n]$, 
  then $(\homDBtoquery\circ h)(w)=\varx$,
\item
  if $\homh(w)=\vertb_\indj$ for some $\indj\in[n]$, 
  then $(\homDBtoquery\circ h)(w)=\vary$,
\item
  if $\homh(w)=d$ for some $d\in\Cons(q)$, 
  then $(\homDBtoquery\circ h)(w)=d$.
\end{itemize}

We define the partition
$\partPfull=\Set{\set{\vertc_2},\ldots,\set{\vertc_{\ell-1}},\setc{\verta_\indi}{\indi\in[\dimn]},\setc{\vertb_\indj}{\indj\in[\dimn]}}$
of $\domn$ and say that a mapping
$\homh$ from $\Vars(q)\cup\Dom$ to $\Dom$
\emph{respects} $\partPfull$,
if for each set from the partition there is exactly one element in the
set $h(\Vars(q))\deff\setc{h(w)}{w\in\Vars(q)}$.
\begin{claim}\label{claim:respectinghom_uMv_testingii}
For every $i\in[n]$, the following are equivalent:
\begin{itemize}
\item
There is a $j\in[n]$ such that $M_{i,j}=1$ and $\vecv_j=1$.
\item
There is a homomorphism $\homh\colon q\to \DB$ that
respects $\partPfull$ such that $a_i\in\homh(\Vars(q))$.
\end{itemize}
\end{claim}
\begin{proof}
Consider a fixed $i\in[n]$.
For one direction assume that there is a $j\in[n]$ such that
$M_{i,j}=1$ and $\vecv_j=1$. 
Then, $\iotasubij$ is a homomorphism from $q$ to $\DB$. Obviously,
$\iotasubij$ respects $\partPfull$, and $a_i=\iotasubij(x)\in\iotasubij(\Vars(q))$.  

For the other direction assume that
$\homh\colon q\to \DB$ is a homomorphism that respects
$\partPfull$ and $a_i\in \homh(\Vars(q))$.  
Thus, there are elements $w_{a_i}, w_b, w_2,\ldots,w_{\ell-1}$ in $\Vars(q)$ such that
$\homh(w_{a_i})=a_i$, $\homh(w_b)\in\setc{b_j}{j\in[n]}$, and $\homh(w_s)=c_s$ for each $s\in\set{2,\ldots,\ell{-}1}$.
It follows that $(\homDBtoquery\circ\homh)$ is a
  bijective homomorphism from $q_\phi(z_1,\ldots,z_k)$ to 
$q_\phi((\homDBtoquery\circ\homh)(z_1),\ldots,(\homDBtoquery\circ\homh)(z_k))$.

  Therefore, 
  it can easily be verified that
  $\homh\circ(\homDBtoquery\circ\homh)^{-1}$ is a homomorphism from
  $q$ to $\DB$ which equals $\iotasubij$ for some
  $\indj\in[\dimn]$.  Thus, for some $j\in[n]$ we have $M_{i,j}=1$ and $\vecv_j=1$.
\end{proof}
\begin{claim}\label{claim:respectinghom_core_testingii}
If  $q$ is a homomorphic core, then every homomorphism $\homh\colon q\to \DB$ respects $\partPfull$. 
\end{claim}
\begin{proof}
  Assume for contradiction that $\homh\colon q\to \DB$ is a
  homomorphism that does not respect $\partPfull$.  Then
  $(\homDBtoquery\circ\homh)$ is a homomorphism from $q$ into a
  proper subquery of $q$, contradicting that $q$
  is a homomorphic core.
\end{proof}

\begin{claim}\label{claim:testing-lowerbound-ii}
If $q$ is a homomorphic core, then for every $i\in[n]$ the following are equivalent:
\begin{itemize}
\item
There is a $j\in[n]$ such that $M_{i,j}=1$ and $\vecv_j=1$.
\item
$(a_i,c_2\ldots,c_k)\in q(\DB)$\,.
\end{itemize}
\end{claim}
\begin{proof}
Consider a fixed $i\in[n]$.
For one direction assume that there is a $j\in[n]$ such that $M_{i,j}=1$ and $\vecv_j=1$.
Then, $\iotasubij$ is a homomorphism from $q$ to $\DB$, and thus
$\big(\iotasubij(x),\iotasubij(z_2),\ldots,\iotasubij(z_k)\big)\in q(\DB)$. 
By definition of $\iotasubij$ we have
 $\big(\iotasubij(x),\iotasubij(z_2),\ldots,\iotasubij(z_k)\big)\allowbreak =(a_i,c_2,\ldots,c_k)$, and hence we are done 
(recall that $y=z_\ell$ and $\ell>k$).

For the other direction assume that $(a_i,c_2,\ldots,c_k)\in q(\DB)$. 
Thus, there exists a homomorphism $\homh:q\to\DB$ such that
$(a_i,c_2,\ldots,c_k)=\big(\homh(x),\homh(z_2),\ldots,\homh(z_k)\big)$.
According to Claim~\ref{claim:respectinghom_core_testingii}, $\homh$ respects $\partPfull$. 
Furthermore, $a_i\in\homh(\Vars(q))$, since $\homh(x)=a_i$.
From Claim~\ref{claim:respectinghom_uMv_testingii},
we obtain that there is a $j\in[n]$ such that $M_{i,j}=1$ and $\vecv_j=1$.
\end{proof}

We are now ready for proving Theorem~\ref{thm:testing}\,\eqref{item:thm:testing:lowerbound} for the case that
$q$ violates condition \eqref{eq:t-hier-cond-ii} 
of Definition~\ref{def:thierarchical}.
  Assume for contradiction that the
  testing problem
  for $q$ can be solved with update time
  $\updatetime=\bigoh(\actdomsize^{1-\smalleps})$ 
  and testing time $\testingtime=\bigoh(\actdomsize^{1-\smalleps})$.
  We can use this
  algorithm to solve the OuMv-problem in time
  $\bigoh(\actdomsize^{3-\smalleps})$ as follows.

  In the
  preprocessing phase, we are given the $\dimn\times\dimn$ matrix
  $\matM$ and let $\vecv^{\,0}$ be the all-zero vectors
  of dimension $\dimn$.
  We start the preprocessing phase of our testing algorithm for
  $q$ with the empty database. As this database has constant size,
  the preprocessing phase finishes in constant time.
  Afterwards, we use $\bigOh(\dimn^2)$ insert operations to build the
  database $\DB(q,\matM,\vecv^{\,0})$.
  All this is done within time $\bigOh(n^2\updatetime)=\bigOh(n^{3-\epsilon})$.

  When a pair
  of vectors $\vecu^{\,\indt}$, $\vecv^{\,\indt}$ (for $t\in[\dimn]$) arrives, we change the
  current database \allowbreak
  $\DB(q,\matM,\vecv^{\,\indt-1})$ into 
  $\DB(q,\matM,\vecv^{\,\indt})$ by using
  at most $\dimn$ update steps.  

  By assumption, $q$ is a homomorphic core. 
  Thus, Claim~\ref{claim:testing-lowerbound-ii} tells us that for 
  $\DB \deff \DB(q,\matM,\vecv^{\,\indt})$
  and for every $i\in[n]$ we have
  \[
    (a_i,c_2,\ldots,c_k)\in q(\DB)
    \quad \iff \quad
    \text{there is a $j\in[n]$ such that $M_{i,j}=1$ and $\vecv_j=1$}\,,
  \]
  
  Hence, after running the $\TEST$ routine with input $(a_i,c_2,\ldots,c_k)$ for
  each $i\in[n]$ with $\vecu_i=1$, 
  we can output the value of $(\vecu^{\,\indt})\trans \matM \vecv^{\,\indt}$.
  For this, we use at most $n$ calls of the $\TEST$ routine, and each such call is executed 
  within time 
  $\testingtime=\bigOh(\setsize{\adom{\DB}}^{1-\smalleps}) \allowbreak = \bigoh(\dimn^{1-\smalleps})$.
  The time we spend to compute $(\vecu^{\,t})\trans M \vecv^{\,t}$ for a fixed $t\in[n]$ is therefore bounded by
  $\bigOh(\dimn \updatetime+\dimn\testingtime)=\bigoh(\dimn^{2-\smalleps})$.
  Thus, the
  overall running time for solving the OuMv-problem sums up to $\bigoh(\dimn^{3-\smalleps})$, contradicting 
  the OuMv-conjecture and hence also the OMv-conjecture.

This completes the proof of Theorem~\ref{thm:testing}\,\eqref{item:thm:testing:lowerbound} for the case that
$q$ violates condition \eqref{eq:t-hier-cond-ii} 
of Definition~\ref{def:thierarchical}.
\qed

\end{document}